\documentclass[journal]{IEEEtran}

\IEEEoverridecommandlockouts
\usepackage{amsfonts}
\usepackage{color}
\usepackage{graphicx}
\usepackage[dvips]{epsfig}
\usepackage{graphics} 
\usepackage{times} 
\usepackage[cmex10]{amsmath} 
\usepackage{amssymb}  
\usepackage{cite}
\usepackage{multirow}
\usepackage[tight,footnotesize]{subfigure}

\def\inftyn #1{\left\|#1\right\|_{\infty}}
\def\twon #1{\left\|#1\right\|_2}
\def\onen #1{\left\|#1\right\|_1}
\def\zeron #1{\left\|#1\right\|_0}
\def\frobn #1{\left\|#1\right\|_{\text{F}}}

\def\abs #1{\left|#1\right|}
\def\inp #1{\left\langle#1\right\rangle}

\def\st{\textrm{ subject to }}

\def\bC{\mathbb{C}}

\def\bR{\mathbb{R}}

\def\m #1{\boldsymbol{#1}}

\def\cC{\mathcal{C}}
\def\cD{\mathcal{D}}

\def\cI{\mathcal{I}}

\def\bee{\begin{equation}}
\def\ene{\end{equation}}

\def\beq{\begin{eqnarray}}
\def\enq{\end{eqnarray}}
\def\lentwo{\setlength\arraycolsep{2pt}}

\newtheorem{lem}{Lemma}
\newtheorem{rem}{Remark}
\newtheorem{cor}{Corollary}
\newtheorem{thm}{Theorem}

\newtheorem{defi}{Definition}

\def\equ #1{\begin{equation}#1\end{equation}}
\def\equa #1{\begin{eqnarray}#1\end{eqnarray}}
\def\sbra #1{\left(#1\right)}
\def\mbra #1{\left[#1\right]}
\def\lbra #1{\left\{#1\right\}}
\def\diag #1{\text{diag}#1}

\def\st {\text{ subject to }}

\title{Robustly Stable Signal Recovery in Compressed Sensing with Structured Matrix Perturbation}
\author{Zai Yang, Cishen Zhang, and Lihua Xie$^*$, \emph{Fellow, IEEE}

\thanks{$^*$Author for correspondence.

Z. Yang and L. Xie are with EXQUISITUS, Centre for E-City, School of Electrical and Electronic Engineering, Nanyang Technological University, 639798, Singapore (e-mail: yang0248@e.ntu.edu.sg; elhxie@ntu.edu.sg).

C. Zhang is with the Faculty of Engineering and Industrial Sciences, Swinburne University of Technology, Hawthorn VIC 3122, Australia (e-mail: cishenzhang@swin.edu.au).}}

\begin{document}
\maketitle

\begin{abstract}
The sparse signal recovery in the standard compressed sensing (CS) problem requires that the sensing matrix be known {\em a priori}. Such an ideal assumption may not be met in practical applications where various errors and fluctuations exist in the sensing instruments. This paper considers the problem of compressed sensing subject to a structured perturbation in the sensing matrix. Under mild conditions, it is shown that a sparse signal can be recovered by $\ell_1$ minimization and the recovery error is at most proportional to the measurement noise level, which is similar to the standard CS result. In the special noise free case, the recovery is exact provided that the signal is sufficiently sparse with respect to the perturbation level. The formulated structured sensing matrix perturbation is applicable to the direction of arrival estimation problem, so has practical relevance.  Algorithms are proposed to implement the $\ell_1$ minimization problem and numerical simulations are carried out to verify the result obtained.
\end{abstract}



\section{Introduction}
Compressed sensing (CS) has been a very active research area since
the pioneering works of Cand{\`e}s \emph{et al.}
\cite{candes2006stable,candes2006robust} and Donoho
\cite{donoho2006compressed}. In CS, a signal $\m{x}^o\in\bR^n$ of
length $n$ is called $k$-sparse if it has at most $k$ nonzero
entries, and it is called compressible if its entries obey a power
law \equ{\abs{x^o}_{\sbra{j}}\leq C_qi^{-q},\label{formu:powerlaw}}
where $\abs{x^o}_{\sbra{j}}$ is the $j$th largest entry (in absolute
value) of $\m{x}^o$
($\abs{x^o}_{\sbra{1}}\geq\abs{x^o}_{\sbra{2}}\geq\cdots\geq\abs{x^o}_{\sbra{n}}$),
$q>1$ and $C_q$ is a constant that depends only on $q$. Let
$\m{x}^k$ be a vector that keeps the $k$ largest
entries (in absolute value) of $\m{x}^o$ with the rest being zeros. If $\m{x}^o$ is
compressible, then it can be well approximated by the sparse signal
$\m{x}^k$ in the sense that \equ{\twon{\m{x}^o-\m{x}^k}\leq C'_q
k^{-q+1/2}\label{formu:2normerror_compressible}} where $C'_q$ is a
constant. To obtain the knowledge of $\m{x}^o$, CS acquires
linear measurements of $\m{x}^o$ as
\equ{\m{y}=\m{\Phi}\m{x}^o+\m{e},\label{formu:observationmodel}} where
$\m{\Phi}\in\bR^{m\times n}$ is the sensing matrix (or linear operator)
with typically $k<m\ll n$, $\m{y}\in\bR^m$ is the vector of
measurements, and $\m{e}\in\bR^m$ denotes the vector of measurement
noises with bounded energy, i.e., $\twon{\m{e}}\leq\epsilon$ for $\epsilon>0$. Given $\m{\Phi}$
and $\epsilon$, the task of CS is to recover $\m{x}^o$ from a
significantly reduced number of measurements $\m{y}$. Cand{\`e}s
\emph{et al.} \cite{candes2006stable,candes2008restricted} show that
if $\m{x}^o$ is sparse, then it can be stably recovered under mild
conditions on $\m{\Phi}$ with the recovery error being at most
proportional to the measurement noise level $\epsilon$ by solving an
$\ell_1$ minimization problem. Similarly, the largest entries (in absolute value) of a compressible signal can be stably recovered. More details are presented in
Subsection \ref{sec:standardCS}. In addition to the $\ell_1$
minimization, other approaches that provide similar guarantees are
also reported thereafter, such as IHT \cite{blumensath2009iterative}
and greedy pursuit methods including OMP\cite{tropp2007signal},
StOMP\cite{donoho2006sparse} and CoSaMP\cite{needell2009cosamp}.

The sensing matrix $\m{\Phi}$ is assumed known {\em a priori}
in standard CS, which is, however, not always the case in practical
situations. For example, a matrix perturbation can be caused by
quantization during implementation. In source separation
\cite{blumensath2007compressed,xu2009compressed} the sensing matrix
(or mixing system) is usually unknown and needs to be estimated, and
thus estimation errors exist. In source localization such as
direction of arrival (DOA) estimation
\cite{zhang2011directions,yang2011off} and radar
imaging\cite{herman2009high,zhang2009achieving}, the sensing matrix
(overcomplete dictionary) is constructed via discretizing one or
more continuous parameters, and errors exist typically in the
sensing matrix since the true source locations may not be exactly on
a discretized sampling grid.

There have been recent
active studies on the CS problem where the sensing matrix is unknown or
subject to an unknown perturbation. Gleichman and Eldar \cite{gleichman2011blind} introduce a concept named as blind CS where the sensing matrix is assumed unknown.\footnote{The CS problem formulation in \cite{gleichman2011blind} is a little different from that in (\ref{formu:observationmodel}). In \cite{gleichman2011blind}, the signal of interest is assumed to be sparse in a sparsity basis while the sparsity basis is absorbed in the sensing matrix $\m{\Phi}$ in our formulation. The sparsity basis is assumed unknown in \cite{gleichman2011blind} that leads to an unknown sensing matrix in our formulation.} In order for the measurements $\m{y}$ to determine a unique sparse solution, three additional constraints on $\m{\Phi}$ are studied individually and sufficient conditions are provided to guarantee the uniqueness. Herman and Strohmer
\cite{herman2010general} analyze the effect of a general matrix
perturbation and show that the signal recovery is robust to the
perturbation in the sense that the recovery error grows linearly
with the perturbation level. Similar robust recovery results are
also reported in \cite{herman2010mixed,chi2011sensitivity}. It is demonstrated in \cite{chae2010effects,chi2011sensitivity} that the signal recovery may suffer from a large error under a large perturbation. In addition, the
existence of recovery error caused by the perturbed sensing matrix
is independent of the sparsity of the original signal. Algorithms
have also been proposed to deal with sensing matrix perturbations.
Zhu \emph{et al.} \cite{zhu2011sparsity} propose a sparse total
least-squares approach to alleviating the effect of perturbation
where they explore the structure of the perturbation to improve
recovery performance. Yang \emph{et al.}\cite{yang2011off} formulate
the off-grid DOA estimation problem from a sparse Bayesian inference
perspective and iteratively recover the source signal and the matrix
perturbation. It is noted that existing algorithmic results
provide no guarantees on signal recovery accuracy when there
exist perturbations in the sensing matrix.

This paper is on the perturbed CS problem. A structured matrix perturbation is studied with each
column of the perturbation matrix being a (unknown) constant times a
(known) vector which defines the direction of perturbation. For certain structured
matrix perturbation, we provide conditions for guaranteed signal recovery
performance. Our analysis shows that robust stability (see definition in Subsection \ref{sec:definition}) can be achieved for a sparse signal under similar mild conditions as those for standard CS problem by solving an $\ell_1$ minimization problem
incorporated with the perturbation structure. In the special noise
free case, the recovery is exact for a sufficiently sparse signal
with respect to the perturbation level. A similar result holds for a compressible signal under an additional assumption of small perturbation (depending on
the number of largest entries to be recovered). A practical application problem, the off-grid DOA estimation, is further considered. It can be formulated into our proposed signal recovery problem subject to the structured sensing matrix perturbation, showing the practical relevance of our proposed problem and solution. To verify the obtained results, two algorithms for positive-valued and general signals respectively are proposed to solve the resulting nonconvex $\ell_1$ minimization problem. Numerical simulations confirm our robustly stable signal recovery results.


A common approach in CS to signal recovery is solving an
optimization problem, e.g., $\ell_1$ minimization. In this connection, another
contribution of this paper is to characterize a set of solutions to
the optimization problem that can be good estimates of the signal to
be recovered, which indicates that it is not necessary to obtain the
optimal solution to the optimization problem. This is helpful to
assess the ``effectiveness'' of an algorithm (see definition in
Subsection \ref{sec:Concept_Effectiveness}), for example, the
$\ell_p$ ($p<1$) minimization
\cite{chartrand2008restricted,saab2008stable} in standard CS, in
solving the optimization problem since in nonconvex optimization the
output of an algorithm cannot be guaranteed to be the optimal
solution.

Notations used in this paper are as follows. Bold-case letters are
reserved for vectors and matrices. $\zeron{\m{x}}$ denotes the
pseudo $\ell_0$ norm that counts the number of nonzero entries of a
vector $\m{x}$. $\onen{\m{x}}$ and $\twon{\m{x}}$ denote the
$\ell_1$ and $\ell_2$ norms of a vector $\m{x}$ respectively.
$\twon{\m{A}}$ and $\frobn{\m{A}}$ are the spectral and Frobenius norms of a matrix $\m{A}$ respectively. $\m{x}^T$ is the transpose of a vector $\m{x}$ and $\m{A}^T$ is for a matrix
$\m{A}$. $x_j$ is the $j$th entry of a vector $\m{x}$.
$T^c$ is the complementary set of a set $T$. Unless otherwise stated, $\m{x}_T$ has entries of a vector $\m{x}$ on an index set $T$ and zero entries on $T^c$. $\diag\sbra{\m{x}}$ is a diagonal matrix
with its diagonal entries being entries of a vector $\m{x}$. $\odot$
is the Hadamard (elementwise) product.

The rest of the paper is organized as follows. Section
\ref{sec:existingresults} first defines formally some terminologies used in this paper and then introduces existing results on standard CS and
perturbed CS. Section \ref{sec:SP-CS} presents the main results
of the paper as well as some discussions and a practical application in DOA estimation. Section
\ref{sec:algorithms} introduces algorithms for the $\ell_1$ minimization problem in our
considered perturbed CS and their analysis. Section \ref{sec:simulation} presents extensive numerical simulations to verify our main results and also empirical results of DOA estimation to support the theoretical findings. Conclusions are drawn in Section \ref{sec:conclusion}. Finally, some mathematical proofs are provided in Appendices.

\section{Preliminary Results}\label{sec:existingresults}

\subsection{Definitions} \label{sec:definition}
For the purpose of clarification of expression, we define formally some terminologies for signal recovery used in this paper, including stability in standard CS, robustness and robust stability in perturbed CS.

\begin{defi}[\cite{candes2006stable}] In standard CS where $\m{\Phi}$ is known {\em a priori}, consider a recovered signal $\widehat{\m{x}}$ of $\m{x}^o$ from measurements $\m{y}=\m{\Phi}\m{x}^o+\m{e}$ with $\twon{\m{e}}\leq\epsilon$. We say that $\widehat{\m{x}}$ achieves stable signal recovery if
\equ{\twon{\widehat{\m{x}}-\m{x}^o}\leq C^{stb}_1k^{-q+1/2}+C^{stb}_2\epsilon\notag}
holds for compressible signal $\m{x}^o$ obeying (\ref{formu:powerlaw}) and an integer $k$, or if
\equ{\twon{\widehat{\m{x}}-\m{x}^o}\leq C^{stb}_2\epsilon\notag}
holds for $k$-sparse signal $\m{x}^o$,
with nonnegative constants $C^{stb}_1$, $C^{stb}_2$.
\end{defi}

\begin{defi} In perturbed CS where $\m{\Phi}=\m{A}+\m{E}$ with $\m{A}$ known {\em a priori} and $\m{E}$ unknown with $\frobn{\m{E}}\leq\eta$, consider a recovered signal $\widehat{\m{x}}$ of $\m{x}^o$ from measurements $\m{y}=\m{\Phi}\m{x}^o+\m{e}$ with $\twon{\m{e}}\leq\epsilon$. We say that $\widehat{\m{x}}$ achieves robust signal recovery if
\equ{\twon{\widehat{\m{x}}-\m{x}^o}\leq C^{rbt}_1k^{-q+1/2}+C^{rbt}_2\epsilon+C^{rbt}_3\eta\notag}
holds for compressible signal $\m{x}^o$ obeying (\ref{formu:powerlaw}) and an integer $k$, or if
\equ{\twon{\widehat{\m{x}}-\m{x}^o}\leq C^{rbt}_2\epsilon+C^{rbt}_3\eta\notag}
holds for $k$-sparse signal $\m{x}^o$,
with nonnegative constants $C^{rbt}_1$, $C^{rbt}_2$ and $C^{rbt}_3$. \label{defi:robust}
\end{defi}

\begin{defi} In perturbed CS where $\m{\Phi}=\m{A}+\m{E}$ with $\m{A}$ known {\em a priori} and $\m{E}$ unknown with $\frobn{\m{E}}\leq\eta$, consider a recovered signal $\widehat{\m{x}}$ of $\m{x}^o$ from measurements $\m{y}=\m{\Phi}\m{x}^o+\m{e}$ with $\twon{\m{e}}\leq\epsilon$. We say that $\widehat{\m{x}}$ achieves robustly stable signal recovery if
\equ{\twon{\widehat{\m{x}}-\m{x}^o}\leq C^{rs}_1\sbra{\eta}k^{-q+1/2}+C^{rs}_2\sbra{\eta}\epsilon\notag}
holds for compressible signal $\m{x}^o$ obeying (\ref{formu:powerlaw}) and an integer $k$, or if
\equ{\twon{\widehat{\m{x}}-\m{x}^o}\leq C^{rs}_2\sbra{\eta}\epsilon\notag}
holds for $k$-sparse signal $\m{x}^o$,
with nonnegative constants $C^{rs}_1$, $C^{rs}_2$ depending on $\eta$. \label{defi:robuststable}
\end{defi}

\begin{rem}~
\begin{itemize}
  \item[(1)] In the case where $\m{x}^o$ is compressible, the defined stable, robust, or robustly stable signal recovery is in fact for its $k$ largest entries (in absolute value). The first term $O\sbra{k^{-q+1/2}}$ in the error bounds above represents, by (\ref{formu:2normerror_compressible}), the best approximation error (up to a scale) that can be achieved when we know everything about $\m{x}^o$ and select its $k$ largest entries.
  \item[(2)] The Frobenius norm of $\m{E}$, $\frobn{\m{E}}$, can be replaced by any other norm  in Definitions \ref{defi:robust} and \ref{defi:robuststable} since the norms are equivalent.
  \item[(3)] By robust stability, we mean that the signal recovery is stable for any fixed matrix perturbation level $\eta$ according to Definition \ref{defi:robuststable}.

\end{itemize}
\end{rem}

It should be noted that the stable recovery in standard CS and the robustly stable recovery in perturbed CS are exact in the noise free, sparse signal case while there is no such a guarantee for the robust recovery in perturbed CS.

\subsection{Stable Signal Recovery of Standard CS}\label{sec:standardCS}
The task of standard CS is to recover the original signal
$\m{x}^o$ via an efficient approach given the sensing matrix
$\m{\Phi}$, acquired sample $\m{y}$ and upper bound $\epsilon$
for the measurement noise. This paper focuses on the $\ell_1$ norm
minimization approach. The restricted isometry property (RIP)
\cite{candes2006compressive} has become a dominant tool to such
analysis, which is defined as follows.

\begin{defi} Define the $k$-restricted isometry constant (RIC) of a matrix $\m{\Phi}$, denoted by   $\delta_k\sbra{\m{\Phi}}$, as the smallest number such that
\equ{\sbra{1-\delta_k\sbra{\m{\Phi}}}\twon{\m{v}}^2\leq\twon{\m{\Phi}\m{v}}^2\leq \sbra{1+\delta_k\sbra{\m{\Phi}}}\twon{\m{v}}^2\notag}
holds for all $k$-sparse vectors $\m{v}$. $\m{\Phi}$ is said to satisfy the $k$-RIP with constant $\delta_k\sbra{\m{\Phi}}$ if $\delta_k\sbra{\m{\Phi}}<1$.
\end{defi}

Based on the RIP, the following theorem holds.

\begin{thm}[\cite{candes2008restricted}] Assume that $\delta_{2k}\sbra{\m{\Phi}}<\sqrt{2}-1$ and $\twon{\m{e}}\leq\epsilon$. Then an optimal solution $\m{x}^*$ to the basis pursuit denoising (BPDN) problem
\equ{\min_{\m{x}}\onen{\m{x}},\st\twon{\m{y}-\m{\Phi}\m{x}}\leq\epsilon \label{formu:BPDN}}
satisfies
\equ{\twon{\m{x}^*-\m{x}^o}\leq C_0^{std}k^{-1/2}\onen{\m{x}^o-\m{x}^k}+C_1^{std}\epsilon \label{formu:ReconUpperBound_BPDN}}
where $C_0^{std}=\frac{2\mbra{1+\sbra{\sqrt{2}-1}\delta_{2k}\sbra{\m{\Phi}}}} {1-\sbra{\sqrt{2}+1}\delta_{2k}\sbra{\m{\Phi}}}$, $C_1^{std}=\frac{4\sqrt{1+\delta_{2k}\sbra{\m{\Phi}}}} {1-\sbra{\sqrt{2}+1}\delta_{2k}\sbra{\m{\Phi}}}$.\label{Thm:BPDN}
\end{thm}

Theorem \ref{Thm:BPDN} states that a $k$-sparse signal $\m{x}^o$ ($\m{x}^k=\m{x}^o$) can be stably recovered by solving a computationally efficient convex optimization problem provided
$\delta_{2k}\sbra{\m{\Phi}}<\sqrt{2}-1$. The same conclusion holds in the case of compressible signal $\m{x}^o$ since
\equ{k^{-1/2}\onen{\m{x}^o-\m{x}^k}\leq C''_qk^{-q+1/2}\label{formu:l1bound}}
according to (\ref{formu:powerlaw}) and (\ref{formu:2normerror_compressible}) with $C''_q$ being a constant.
In the special noise free, $k$-sparse signal case,
such a recovery is exact. The RIP condition in Theorem
\ref{Thm:BPDN} can be satisfied provided $m\geq
O\sbra{k\log\sbra{n/k}}$ with a large probability if the sensing
matrix $\m{\Phi}$ is i.i.d. subgaussian distributed
\cite{baraniuk2008simple}. Note that the RIP
condition for the stable signal recovery in standard CS has been
relaxed in \cite{foucart2009sparsest,cai2010shifting} but it is beyond
the scope of this paper.

\subsection{Robust Signal Recovery in Perturbed CS}\label{sec:robustrecovery}
In standard CS, the sensing matrix $\m{\Phi}$ is assumed to be
exactly known. Such an ideal assumption is not always the case in
practice. Consider that the true sensing matrix is
$\m{\Phi}=\m{A}+\m{E}$ where $\m{A}\in\bR^{m\times n}$ is the known
nominal sensing matrix and $\m{E}\in\bR^{m\times n}$ represents the
unknown matrix perturbation. Unlike the additive noise term $\m{e}$
in the observation model in (\ref{formu:observationmodel}), a
multiplicative ``noise'' $\m{E}\m{x}^o$ is introduced in perturbed CS and is more difficult to analyze since it is correlated
with the signal of interest. Denote $\twon{\m{E}}^{\sbra{k}}$ the
largest spectral norm taken over all $k$-column submatrices of
$\m{E}$, and similarly define $\twon{\m{\Phi}}^{\sbra{k}}$. The
following theorem is stated in \cite{herman2010general}.

\begin{thm}[\cite{herman2010general}] Assume that there exist constants $\varepsilon_{\m{E},\m{\Phi}}^{\sbra{k}}$, $\epsilon$ and $\epsilon_{\m{E},\m{x}^o}$ such that $\frac{\twon{\m{E}}^{\sbra{k}}} {\twon{\m{\Phi}}^{\sbra{k}}}\leq\varepsilon_{\m{E},\m{\Phi}}^{\sbra{k}}$, $\twon{\m{e}}\leq\epsilon$ and $\twon{\m{E}\m{x}^o}\leq\epsilon_{\m{E},\m{x}^o}$. Assume that $\delta_{2k}\sbra{\m{\Phi}}<\frac{\sqrt{2}}{\sbra{1+\varepsilon_{\m{E},\m{\Phi}}^ {\sbra{2k}}}^2}-1$ and $\zeron{\m{x}^o}\leq k$. Then an optimal solution $\m{x}^*$ to the BPDN problem with the nominal sensing matrix $\m{A}$, denoted by N-BPDN,
\equ{\min_{\m{x}}\onen{\m{x}},\st \twon{\m{y}-\m{A}\m{x}}\leq\epsilon+\epsilon_{\m{E},\m{x}^o} \label{formu:N-BPDN}}
achieves robust signal recovery with
\equ{\twon{\m{x}^*-\m{x}^o}\leq C^{ptb}\epsilon + C^{ptb}\epsilon_{\m{E},\m{x}^o}}
where $C^{ptb}=\frac{4\sqrt{1+\delta_{2k}\sbra{\m{\Phi}}}\sbra{1+\varepsilon_{\m{E},\m{\Phi}}^ {\sbra{2k}}} }{1-\sbra{\sqrt{2}+1}\mbra{\sbra{1+\delta_{2k}\sbra{\m{\Phi}}}\sbra{1+\varepsilon_{\m{E},\m{\Phi}}^ {\sbra{2k}}}^2-1}}$.\label{Thm:robust_perturbedCS}
\end{thm}

\begin{rem}~
\begin{itemize}
 \item[(1)] The relaxation of the inequality constraint in (\ref{formu:N-BPDN}) from $\epsilon$ to $\epsilon+\epsilon_{\m{E},\m{x}^o}$ is to ensure that the original signal $\m{x}^o$ is a feasible solution to N-BPDN. Theorem \ref{Thm:robust_perturbedCS} is a little different from that in \cite{herman2010general}, where the multiplicative ``noise'' $\m{E}\m{x}^o$ is bounded using $\varepsilon_{\m{E},\m{\Phi}}^{\sbra{k}}$, $\delta_{k}\sbra{\m{\Phi}}$ and $\twon{\m{\Phi}\m{x}^o}$ rather than a constant $\epsilon_{\m{E},\m{x}^o}$.
 \item[(2)] Theorem \ref{Thm:robust_perturbedCS} is applicable only to the small perturbation case where $\varepsilon_{\m{E},\m{\Phi}}^{\sbra{2k}}<\sqrt[4]{2}-1$ since $\delta_{2k}\sbra{\m{\Phi}}\geq0$.
 \item[(3)] Theorem \ref{Thm:robust_perturbedCS} generalizes Theorem \ref{Thm:BPDN} for the $k$-sparse signal case. As the perturbation $\m{E}\rightarrow0$, Theorem \ref{Thm:robust_perturbedCS} coincides with Theorem \ref{Thm:BPDN} for the $k$-sparse signal case.
\end{itemize}\label{remark:robustrecovery}
\end{rem}

Theorem \ref{Thm:robust_perturbedCS} states that, for a small matrix perturbation $\m{E}$, the signal recovery of N-BPDN that is based on the nominal sensing matrix $\m{A}$ is robust to the perturbation with the recovery error growing at most linearly with the perturbation level. Note that, in general, the signal recovery in Theorem \ref{Thm:robust_perturbedCS} is unstable according to the definition of stability in this paper since the recovery error cannot be bounded within a constant (independent of the noise) times the noise level as some perturbation occurs. A result on general signals in \cite{herman2010general} is omitted that shows the robust recovery of a compressible signal. The same problem is studied and similar results are reported in \cite{herman2010mixed} based on the greedy algorithm CoSaMP \cite{needell2009cosamp}.

\section{SP-CS: CS Subject to Structured Perturbation} \label{sec:SP-CS}
\subsection{Problem Description}
In this paper we consider a structured perturbation in the form
$\m{E}=\m{B}\m{\Delta}^o$ where $\m{B}\in\bR^{m\times n}$ is known {\em
a priori}, $\m{\Delta}^o=\diag\sbra{\m{\beta}^o}$ is a bounded uncertain
term with $\m{\beta}^o\in\mbra{-r,r}^n$ and $r>0$, i.e., each column of
the perturbation is on a known direction. In addition, we assume
that each column of $\m{B}$ has unit norm to avoid the scaling
problem between $\m{B}$ and $\m{\Delta}^o$ (in fact, the D-RIP
condition on matrix $\mbra{\m{A},\m{B}}$ in Subsection
\ref{sec:mainresults} implies that columns of both $\m{A}$ and
$\m{B}$ have approximately unit norms). As a result, the observation
model in (\ref{formu:observationmodel}) becomes
\equ{\m{y}=\m{\Phi}\m{x}^o+\m{e},\quad\m{\Phi}=\m{A}+\m{B}\m{\Delta}^o
\label{formu:observationmodel_perturbed}} with
$\m{\Delta}^o=\diag\sbra{\m{\beta}^o}$, $\m{\beta}^o\in\mbra{-r,r}^n$ and
$\twon{\m{e}}\leq\epsilon$. Given $\m{y}$, $\m{A}$, $\m{B}$, $r$ and
$\epsilon$, the task of SP-CS is to recover $\m{x}^o$ and possibly $\m{\beta}^o$ as well.

\begin{rem}~
 \begin{itemize}
  \item[(1)] Without loss of generality, we assume that $\m{x}$, $\m{y}$, $\m{A}$, $\m{B}$ and $\m{e}$ are all in the real domain unless otherwise stated.
  \item[(2)] If $x^o_j=0$ for some $j\in\lbra{1,\cdots,n}$, then $\beta^o_j$ has no contributions to the observation $\m{y}$ and hence it is impossible to recover $\beta^o_j$. As a result, the recovery of $\m{\beta}^o$ in this paper refers only to the recovery on the support of $\m{x}^o$.
 \end{itemize}
\end{rem}

\subsection{Main Results of This Paper}\label{sec:mainresults}
In this paper, a vector $\m{v}$ is called $2k$-duplicately (D-)
sparse if $\m{v}=\mbra{\m{v}_1^T,\m{v}_2^T}^T$ with $\m{v}_1$ and
$\m{v}_2$ being of the same dimension and jointly $k$-sparse (each
being $k$-sparse and sharing the same support). The concept of
duplicate (D-) RIP is defined as follows.

\begin{defi} Define the $2k$-duplicate (D-) RIC of a matrix $\m{\Phi}$, denoted by $\bar{\delta}_{2k}\sbra{\m{\Phi}}$, as the smallest number such that
\equ{\sbra{1-\bar{\delta}_{2k}\sbra{\m{\Phi}}}\twon{\m{v}}^2\leq\twon{\m{\Phi}\m{v}}^2\leq \sbra{1+\bar{\delta}_{2k}\sbra{\m{\Phi}}}\twon{\m{v}}^2 \notag}
holds for all $2k$-D-sparse vectors $\m{v}$. $\m{\Phi}$ is said to satisfy the $2k$-D-RIP with constant $\bar{\delta}_{2k}\sbra{\m{\Phi}}$ if $\bar{\delta}_{2k}\sbra{\m{\Phi}}<1$.
\end{defi}

With respect to the perturbed observation model in (\ref{formu:observationmodel_perturbed}), let $\m{\Psi}=\mbra{\m{A},\m{B}}$. The main results of this paper are stated in the following theorems. The proof of Theorem \ref{Thm:l0_offgrid} is provided in Appendix A and proofs of Theorems \ref{Thm:BPDN_offgrid} and \ref{Thm:BPDN_offgrid_compressible} are in Appendix B.

\begin{thm} In the noise free case where $\m{e}=\m{0}$, assume that $\zeron{\m{x}^o}\leq k$ and $\bar{\delta}_{4k}\sbra{\m{\Psi}}<1$. Then an optimal solution $\sbra{\m{x}^*,\m{\beta}^*}$ to the perturbed combinatorial optimization problem
\equ{\min_{\m{x}\in\bR^n,\m{\beta}\in\mbra{-r,r}^n}\zeron{\m{x}}, \st \m{y}=\sbra{\m{A}+\m{B}\m{\Delta}}\m{x} \label{formu:l0_offgrid_Phiv}}
with $\m{\Delta}=\diag\sbra{\m{\beta}}$ recovers $\m{x}^o$ and $\m{\beta}^o$. \label{Thm:l0_offgrid}
\end{thm}

\begin{thm} Assume that $\bar{\delta}_{4k}\sbra{\m{\Psi}}<\sbra{\sqrt{2\sbra{1+r^2}}+1}^{-1}$, $\zeron{\m{x}^o}\leq k$ and $\twon{\m{e}}\leq\epsilon$. Then an optimal solution $\sbra{\m{x}^*,\m{\beta}^*}$ to the perturbed (P-) BPDN problem
\equ{\min_{\m{x}\in{\bR}^n,\m{\beta}\in\mbra{-r,r}^n}\onen{\m{x}}, \st \twon{\m{y}-\sbra{\m{A}+\m{B}\m{\Delta}}\m{x}}\leq\epsilon \label{formu:l1_offgrid_Phiv}}
achieves robustly stable signal recovery with
{\lentwo\equa{
&&\twon{\m{x}^*-\m{x}^o}\leq C\epsilon, \label{formu:boundx}\\
&&\twon{\sbra{\m{\beta}^*-\m{\beta}^o} \odot \m{x}^o}\leq \cC\epsilon \label{formu:boundbetax}}
}where
{\lentwo\equa{C
&=& \frac{4\sqrt{1+\bar\delta_{4k}\sbra{\m{\Psi}}}} {1-\sbra{\sqrt{2\sbra{1+r^2}}+1}\bar{\delta}_{4k}\sbra{\m{\Psi}}},\notag\\\cC
&=& \frac{\mbra{2+\sqrt{1+r^2}\twon{\m{\Psi}}C}} {\sqrt{1-\bar\delta_{4k}\sbra{\m{\Psi}}}}.\notag }}
\label{Thm:BPDN_offgrid}
\end{thm}

\begin{thm} Assume that $\bar{\delta}_{4k}\sbra{\m{\Psi}}<\sbra{\sqrt{2\sbra{1+r^2}}+1}^{-1}$ and $\twon{\m{e}}\leq\epsilon$. Then an optimal solution $\sbra{\m{x}^*,\m{\beta}^*}$ to the P-BPDN problem in (\ref{formu:l1_offgrid_Phiv}) satisfies that
{\lentwo\equa{
&&\twon{\m{x}^*-\m{x}^o}\leq \sbra{C_0k^{-1/2} + C_1}\onen{\m{x}^o-\m{x}^k}+C_2\epsilon, \label{formu:boundx}\\
&&\twon{\sbra{\m{\beta}^*-\m{\beta}^o} \odot \m{x}^k}\notag\\
&&\qquad\qquad\text{ }\leq\sbra{\cC_0k^{-1/2} + \cC_1}\onen{\m{x}^o-\m{x}^k}+\cC_2\epsilon \label{formu:boundbetax}}
}where {\lentwo\equa{C_0
&=& 2\mbra{1+\sbra{\sqrt{2\sbra{1+r^2}}-1}\bar{\delta}_{4k}\sbra{\m{\Psi}}}/a,\notag\\ C_1
&=& 2\sqrt{2}r\bar\delta_{4k}\sbra{\m{\Psi}}/a,\notag\\ \cC_0
&=& \sqrt{1+r^2}\twon{\m{\Psi}}C_0/b,\notag\\ \cC_1
&=& \mbra{\sqrt{1+r^2}C_1+2r}\twon{\m{\Psi}}/b\notag }
}with $a=1-\sbra{\sqrt{2\sbra{1+r^2}}+1}\bar{\delta}_{4k}\sbra{\m{\Psi}}$, $b=\sqrt{1-\bar\delta_{4k}\sbra{\m{\Psi}}}$ and $C_2=C$, $\cC_2=\cC$ with $C,\cC$ as defined in Theorem \ref{Thm:BPDN_offgrid}. \label{Thm:BPDN_offgrid_compressible}
\end{thm}

\begin{rem} In general, the robustly stable signal recovery cannot be concluded for compressible signals since the error bound in (\ref{formu:boundx}) may be very large in the case of large perturbation by $C_1=O\sbra{r}$.
If the perturbation is small with $r=O\sbra{k^{-1/2}}$, then the robust stability can be achieved  for compressible signals by (\ref{formu:l1bound}) provided that the D-RIP condition in Theorem \ref{Thm:BPDN_offgrid_compressible} is satisfied. \label{rem:robuststable}
\end{rem}


\subsection{Interpretation of the Main Results}\label{sec:discussion}
Theorem \ref{Thm:l0_offgrid} states that for a $k$-sparse signal
$\m{x}^o$, it can be recovered by solving a combinatorial
optimization problem provided $\bar{\delta}_{4k}\sbra{\m{\Psi}}<1$ when
the measurements are exact. Meanwhile, $\m{\beta}^o$ can be recovered.
Since the combinatorial optimization problem is NP-hard and that its
solution is sensitive to measurement noise\cite{donoho2006stable}, a
more reliable approach, $\ell_1$ minimization, is explored in
Theorems \ref{Thm:BPDN_offgrid} and
\ref{Thm:BPDN_offgrid_compressible}.

Theorem \ref{Thm:BPDN_offgrid} states the robustly stable recovery of a $k$-sparse
signal $\m{x}^o$ in SP-CS with the recovery
error being at most proportional to the noise level. Such robust stability is obtained by solving an
$\ell_1$ minimization problem incorporated with the perturbation
structure provided that the D-RIC is sufficiently small with
respect to the perturbation level in terms of $r$. Meanwhile, the
perturbation parameter $\m{\beta}^o$ can be stably recovered on the
support of $\m{x}^o$. As the D-RIP condition is satisfied in Theorem
\ref{Thm:BPDN_offgrid}, the signal recovery error of perturbed
CS is constrained by the noise level $\epsilon$, and the influence
of the perturbation is limited to the coefficient before $\epsilon$.
For example, if $\bar\delta_{4k}\sbra{\m{\Psi}}=0.2$, then
$\twon{\m{x}^*-\m{x}^o}\leq8.48\epsilon, 8.50\epsilon, 11.0\epsilon$
corresponding to $r=0.01, 0.1, 1$, respectively. In the special
noise free case, the recovery is exact. This is similar to that in
standard CS but in contrast to the existing robust signal recovery
result in Subsection \ref{sec:robustrecovery} where the recovery
error exists once a matrix perturbation appears. Another interpretation of the D-RIP
condition in Theorem \ref{Thm:BPDN_offgrid} is that the robustly stable
signal recovery requires that
$r<\sqrt{\frac{1}{2}\sbra{{\bar\delta}_{4k}\sbra{\m{\Psi}}^{-1}-1}^2-1}$
for a fixed matrix $\m{\Psi}$. Using the aforementioned example where
$\bar\delta_{4k}\sbra{\m{\Psi}}=0.2$, the perturbation is required to
satisfy $r<\sqrt{7}$. As a result, our robustly
stable signal recovery result of SP-CS applies to the case of
large perturbation if the D-RIC of $\m{\Psi}$ is sufficiently small
while the existing result does not as demonstrated in Remark
\ref{remark:robustrecovery}.

Theorem \ref{Thm:BPDN_offgrid_compressible} considers general
signals and is a generalized form of Theorem \ref{Thm:BPDN_offgrid}.
In comparison with Theorem \ref{Thm:BPDN} in standard CS, one
more term $C_1\onen{\m{x}^o-\m{x}^k}$ appears in the upper bound of
the recovery error. The robust stability does not hold generally for compressible signals as illustrated in Remark \ref{rem:robuststable} while it is true under an additional assumption $r=O\sbra{k^{-1/2}}$.

The results in this paper generalize that in standard CS.
Without accounting for the symbolic difference between $\delta_{2k}\sbra{\m{\Phi}}$ and
$\bar\delta_{4k}\sbra{\m{\Psi}}$, the conditions in Theorems
\ref{Thm:BPDN} and \ref{Thm:BPDN_offgrid_compressible} coincide, as
well as the upper bounds in (\ref{formu:ReconUpperBound_BPDN}) and
(\ref{formu:boundx}) for the recovery errors, as the perturbation
vanishes or equivalently $r\rightarrow0$. As mentioned before, the
RIP condition for guaranteed stable recovery in standard CS has
been relaxed. Similar techniques may be adopted to possibly relax
the D-RIP condition in SP-CS. While this paper is focused on the
$\ell_1$ minimization approach, it is also possible to modify other
algorithms in standard CS and apply them to SP-CS to
provide similar recovery guarantees.

\subsection{When is the D-RIP satisfied?}

Existing works studying the RIP mainly focus on random matrices. In
standard CS, $\m{\Phi}$ has the $k$-RIP with constant $\delta$ with
a large probability provided that $m\geq C_{\delta}k\log\sbra{n/k}$
and $\m{\Phi}$ has properly scaled i.i.d. subgaussian distributed
entries with constant $C_{\delta}$ depending on $\delta$ and the
distribution\cite{baraniuk2008simple}. The D-RIP can be considered
as a model-based RIP introduced in
\cite{baraniuk2010model}. Suppose that $\m{A}$, $\m{B}$ are mutually
independent and both are i.i.d. subgaussian distributed (the true
sensing matrix $\m{\Phi}=\m{A}+\m{B}\m{\Delta}^o$ is also i.i.d. subgaussian
distributed if $\m{\beta}^o$ is independent of $\m{A}$ and $\m{B}$). The
model-based RIP is determined by the number of subspaces of the
structured sparse signals that are referred to as the D-sparse ones
in the present paper. For $\m{\Psi}=\mbra{\m{A},\m{B}}$, the number of
$2k$-dimensional subspaces for $2k$-D-sparse signals is
$\begin{pmatrix}n\\k\end{pmatrix}$. Consequently, $\m{\Psi}$ has the
$2k$-D-RIP with constant $\delta$ with a large probability also
provided that $m\geq C_{\delta} k\log\sbra{n/k}$ by \cite[Theorem
1]{baraniuk2010model} or \cite[Theorem 3.3]{blumensath2009sampling}.
So, in the case of a high dimensional system and $r\rightarrow0$, the D-RIP condition on
$\m{\Psi}$ in Theorem \ref{Thm:BPDN_offgrid} or
\ref{Thm:BPDN_offgrid_compressible} can be satisfied when the RIP condition on $\m{\Phi}$
(after proper scaling of its columns) in standard CS is met. It
means that the perturbation in SP-CS gradually strengthens the D-RIP
condition for robustly stable signal recovery but there exists no gap between
SP-CS and standard CS in the case of high dimensional systems.

It is noted that there is another way to stably recover the original
signal $\m{x}^o$ in SP-CS. Given the sparse signal case as an
example where $\m{x}^o$ is $k$-sparse. Let
$\m{z}^o=\begin{bmatrix}\m{x}^o\\\m{\beta}^o\odot\m{x}^o\end{bmatrix}$,
and it is $2k$-sparse. The observation model can be written as
$\m{y}=\m{\Psi}\m{z}^o+\m{e}$. Then $\m{z}^o$ and hence, $\m{x}^o$, can
be stably recovered from the problem\footnote{It is hard to
incorporate the knowledge $\m{\beta}^o\in\mbra{-r,r}^n$ into the problem
in (\ref{formu:TPS-BPDN}).}
\equ{\min_{\m{z}}\onen{\m{z}},\st\twon{\m{y}-\m{\Psi}\m{z}}\leq\epsilon
\label{formu:TPS-BPDN}} provided that
$\delta_{4k}\sbra{\m{\Psi}}<\sqrt{2}-1$ by Theorem \ref{Thm:BPDN}. It
looks like that we transformed the perturbation into a signal of
interest. Denote TPS-BPDN the problem in (\ref{formu:TPS-BPDN}). In
a high dimensional system, the condition
$\delta_{4k}\sbra{\m{\Psi}}<\sqrt{2}-1$ requires about twice as many as the measurements that makes the D-RIP condition
$\bar\delta_{4k}\sbra{\m{\Psi}}<\sqrt{2}-1$ hold by \cite[Theorem
1]{baraniuk2010model} corresponding to the D-RIP condition in
Theorem \ref{Thm:BPDN_offgrid} or
\ref{Thm:BPDN_offgrid_compressible} as $r\rightarrow0$.
As a result, for a considerable range of perturbation level, the
D-RIP condition in Theorem \ref{Thm:BPDN_offgrid} or
\ref{Thm:BPDN_offgrid_compressible} for P-BPDN is weaker than that
for TPS-BPDN since it varies slowly for a moderate perturbation (as
an example, $\bar\delta_{4k}\sbra{\m{\Psi}}<0.414, 0.413, 0.409$
corresponds to $r=0,0.1,0.2$ respectively). Numerical simulations in
Subsection \ref{sec:simulation} can verify our conclusion.

\subsection{Relaxation of the Optimal Solution}\label{sec:Concept_Effectiveness}
In Theorem \ref{Thm:BPDN_offgrid_compressible} (Theorem
\ref{Thm:BPDN_offgrid} is a special case), $\sbra{\m{x}^*,\m{\beta}^*}$
is required to be an optimal solution to P-BPDN. Naturally, we would
like to know if the requirement of the optimality is necessary for a
``good'' recovery in the sense that a good recovery validates the
error bounds in (\ref{formu:boundx}) and (\ref{formu:boundbetax})
under the conditions in Theorem \ref{Thm:BPDN_offgrid_compressible}.
Generally speaking, the answer is negative since, regarding the
optimality of $\sbra{\m{x}^*,\m{\beta}^*}$, only
$\onen{\m{x}^*}\leq\onen{\m{x}^o}$ and the feasibility of
$\sbra{\m{x}^*,\m{\beta}^*}$ are used in the proof of Theorem
\ref{Thm:BPDN_offgrid_compressible} in Appendix B. Denote $\cD$ the
feasible domain of P-BPDN, i.e., \equ{\begin{split}\cD=\{
&\sbra{\m{x},\m{\beta}}: \m{\beta}\in\mbra{-r,r}^n, \\
&\twon{\m{y}-\sbra{\m{A}+\m{B}\m{\Delta}}\m{x}}\leq\epsilon \text{ with }\m{\Delta}=\diag\sbra{\m{\beta}}\}.\end{split}}
We have the following corollary.

\begin{cor} Under the assumptions in Theorem \ref{Thm:BPDN_offgrid_compressible}, any $\sbra{\m{x},\m{\beta}}\in\cD$ that meets $\onen{\m{x}}\leq\onen{\m{x}^o}$ satisfies that
{\lentwo\equa{
&&\twon{\m{x}-\m{x}^o}\leq \sbra{C_0k^{-1/2} + C_1}\onen{\m{x}^o-\m{x}^k}+C_2\epsilon, \notag\\
&&\twon{\sbra{\m{\beta}-\m{\beta}^o} \odot \m{x}^k}\notag\\
&&\qquad\qquad\text{ }\leq\sbra{\cC_0k^{-1/2} + \cC_1}\onen{\m{x}^o-\m{x}^k}+\cC_2\epsilon \notag}
}with $C_j, \cC_j$, $j=0,1,2$, as defined in Theorem \ref{Thm:BPDN_offgrid_compressible}.
\label{cor:optimalnotnecessary_offgrid}\end{cor}

Corollary \ref{cor:optimalnotnecessary_offgrid} generalizes
Theorem \ref{Thm:BPDN_offgrid_compressible} and its proof follows
directly from that of Theorem \ref{Thm:BPDN_offgrid_compressible}. It
shows that a good recovery in SP-CS is not necessarily an optimal
solution to P-BPDN. A similar result holds in standard CS that generalizes Theorem \ref{Thm:BPDN}, and the proof of
Theorem \ref{Thm:BPDN} in \cite{candes2008restricted} applies
directly to such case.

\begin{cor} Under the assumptions in Theorem \ref{Thm:BPDN}, any $\m{x}$
that meets $\twon{\m{y}-\m{A}\m{x}}\leq\epsilon$
and $\onen{\m{x}}\leq\onen{\m{x}^o}$ satisfies that
\equ{\twon{\m{x}-\m{x}^o}\leq C^{std}_0k^{-1/2}\onen{\m{x}^o-\m{x}^k}+C^{std}_1\epsilon}
with $C^{std}_0, C^{std}_1$ as defined in Theorem \ref{Thm:BPDN}.\label{cor:optimalnotnecessary}
\end{cor}

\begin{figure}
  \centering
  \includegraphics[width=2.2in]{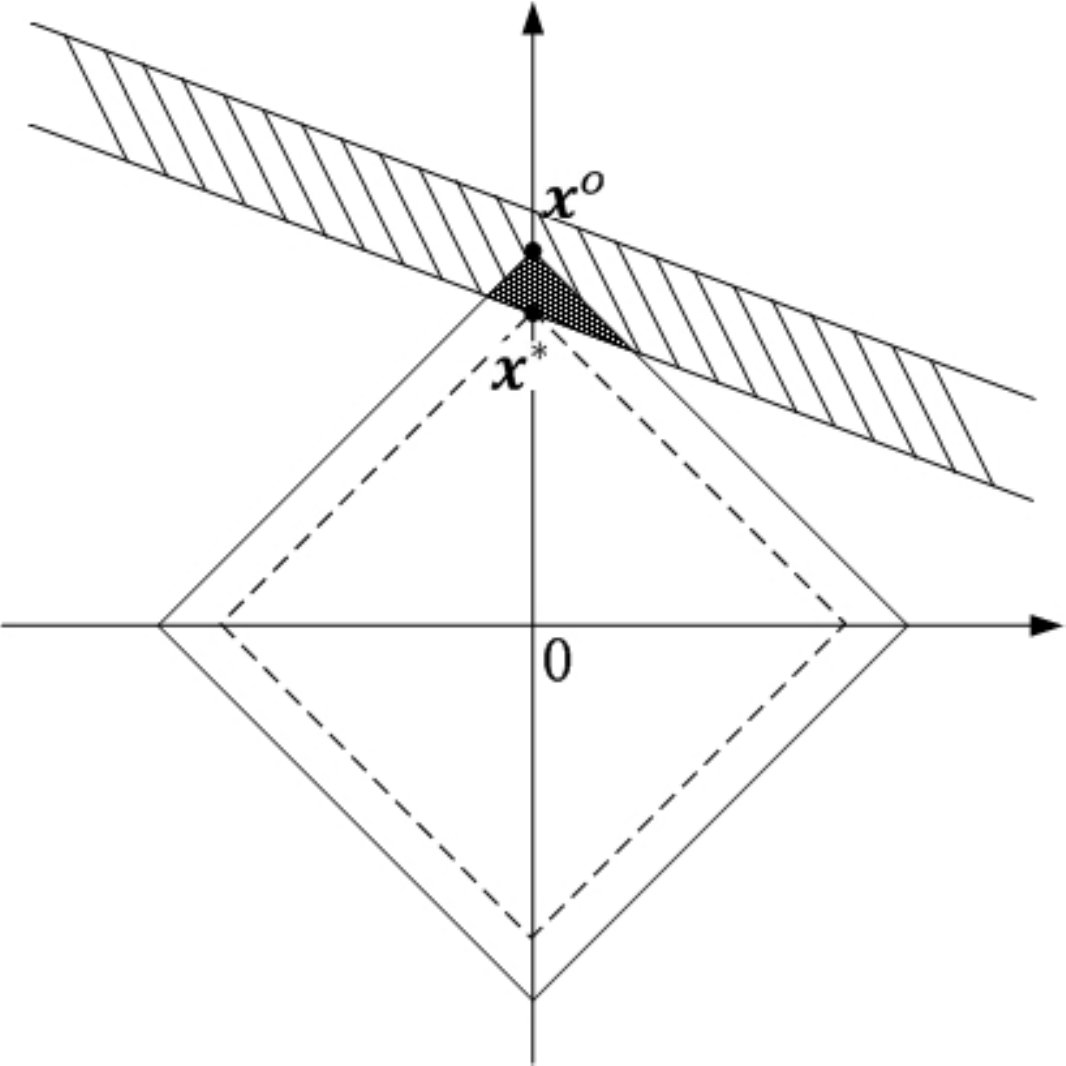}
  \caption{Illustration of Corollary \ref{cor:optimalnotnecessary}. The shaded band area refers to the
  feasible domain of BPDN. The triangular area, the intersection of the feasible domain and the $\ell_1$ ball $\lbra{\m{x}:\onen{\m{x}}\leq\onen{\m{x}^o}}$, is the set of all good recoveries.}
  \label{Fig:illu_solution_BPDN}
\end{figure}

An illustration of Corollary \ref{cor:optimalnotnecessary} is
presented in Fig. \ref{Fig:illu_solution_BPDN}, where the shaded
band area refers to the feasible domain of BPDN in
(\ref{formu:BPDN}) and all points in the triangular area, the
intersection of the feasible domain of BPDN and the $\ell_1$ ball
$\lbra{\m{x}:\onen{\m{x}}\leq\onen{\m{x}^o}}$, are good candidates
for recovery of $\m{x}^o$. The reason why one seeks for the optimal
solution $\m{x}^*$ is to guarantee that the inequality
$\onen{\m{x}}\leq\onen{\m{x}^o}$ holds since $\onen{\m{x}^o}$ is
generally unavailable {\em a priori.} Corollary
\ref{cor:optimalnotnecessary} can explain why a satisfactory
recovery can be obtained in practice using some algorithm that may
not produce an optimal solution to BPDN, e.g., rONE-L1
\cite{yang2010orthonormal}. Corollaries
\ref{cor:optimalnotnecessary_offgrid} and
\ref{cor:optimalnotnecessary} are useful for checking the
effectiveness of an algorithm in the case when the output cannot be
guaranteed to be optimal.\footnote{It is common when the problem to
be solved is nonconvex, such as P-BPDN as discussed in Section
\ref{sec:algorithms} and $\ell_p$ ($0\leq p<1$)
minimization approaches
\cite{chartrand2007exact,chartrand2008restricted, saab2008stable} in
standard CS. In addition, Corollaries
\ref{cor:optimalnotnecessary_offgrid} and
\ref{cor:optimalnotnecessary} can be readily extended to the
$\ell_p$ ($0\leq p<1$) minimization approaches.} Namely, an
algorithm is called effective in solving some $\ell_1$
minimization problem if it can produce a feasible solution with its
$\ell_1$ norm no larger than that of the original signal. Similar
ideas have been adopted in \cite{yang2011sparsity,yang2011phase}.

\subsection{Application to DOA Estimation} \label{sec:DOAapplication}
DOA estimation is a classical problem in signal processing with many practical applications. Its research has recently been advanced owing to the development of CS based methods, e.g., $\ell_1$-SVD\cite{malioutov2005sparse}. This subsection shows that the proposed SP-CS framework is applicable to the DOA estimation problem and hence has practical relevance. Consider $k$ narrowband far-field sources $s_j$, $j=1,\cdots,k$, impinging on an $m$-element uniform linear array (ULA) from directions $d_j$ with $d_j\in[0,\pi)$, $j=1,\cdots,k$. Denote $\m{\theta}=\mbra{\cos\sbra{d_1},\dots,\cos\sbra{d_k}}^T\in(-1,1]^k$. For convenience, we consider the estimation of $\m{\theta}$ rather than that of $\m{d}$ hereafter. Moreover, we consider the noise free case for simplicity of exposition. Then the observation model is $\m{y} = \m{A}\sbra{\m{\theta}}\m{s}$ according to \cite{krim1996two}, where $\m{y}\in\bC^m$ denotes the vector of sensor measurements and $\m{A}\sbra{\m{\theta}}\in\bC^{m\times k}$ denotes the sensing/measurement matrix with respect to $\m{\theta}$ with $\m{A}_{j}\sbra{\m{\theta}}=\m{a}\sbra{\theta_j}$ and $a_l\sbra{\theta_j}=\frac{1}{\sqrt{m}}\exp\lbra{i\pi\sbra{l-\frac{m+1}{2}}\theta_j}$, $j=1,\cdots,k$, $l=1,\cdots,m$, where $i=\sqrt{-1}$. The objective of DOA estimation is to estimate $\m{\theta}$ given $\m{y}$ and possibly $k$ as well. Since $k$ is typically small, CS based methods have been motivated in recent years. Let $\tilde{\m{\theta}}=\lbra{\frac{1}{n}-1, \frac{3}{n}-1, \cdots, 1-\frac{1}{n}}$ be a uniform sampling grid in $\m{\theta}$ range $(-1,1]$ where $n$ denotes the grid number (without loss of generality, we assume that $n$ is an even number). In existing standard CS based methods $\tilde{\m{\theta}}$ actually serves as the set of all candidate DOA estimates. As a result, their estimation accuracy is limited by the grid density since for some $\theta_j$, $j\in\lbra{1,\cdots,k}$, the best estimate of $\theta_j$ is its nearest grid point in $\tilde{\m{\theta}}$. It can be easily shown that a lower bound for the mean squared estimation error of each $\theta_j$ is $\text{LB}=\frac{1}{3n^2}$ by assuming that $\theta_j$ is uniformly distributed in one or more grid intervals.

An off-grid model has been studied in \cite{zhu2011sparsity,yang2011off} that takes into account effects of the off-grid DOAs and introduces a structured matrix perturbation in the measurement matrix. For completeness, we re-derive it using Taylor expansion. Suppose $\theta_j\notin\tilde{\m{\theta}}$ for some $j\in\lbra{1,\cdots,k}$ and that $\tilde{\theta}_{l_j}$, $l_j\in\lbra{1, \cdots, n}$, is the nearest grid point to $\theta_j$. By Taylor expansion we have
\equ{\m{a}\sbra{\theta_j}=\m{a}\sbra{\tilde{\theta}_{l_j}}+\m{b}\sbra{\tilde{\theta}_{l_j}} \sbra{\theta_{j} -\tilde{\theta}_{l_j}} + \m{R}_j \label{formu:offgridlinearization}}
with $\m{b}\sbra{\tilde{\theta}_{l_j}}=\m{a}'\sbra{\tilde{\theta}_{l_j}}$ and $\m{R}_j$ being a remainder term with respect to $\theta_j$. Denote $\kappa=\frac{\pi}{2}\sqrt{\frac{m^2-1}{3}}$, $\m{A}=\mbra{\m{a}\sbra{\tilde{\theta}_1},\cdots,\m{a}\sbra{\tilde{\theta}_n}}$, $\m{B}={\kappa}^{-1}\mbra{\m{b}\sbra{\tilde{\theta}_1},\cdots,\m{b}\sbra{\tilde{\theta}_n}}$, and for $l=1,\cdots,n$,
\[\begin{array}{lll}\beta_l^o=\kappa\sbra{\theta_j-\tilde{\theta}_{l_j}},& x_l^o=s_j,& \text{if }l=l_j \text{ for any } j\in\lbra{1,\cdots,k};\\\beta_l^o=0,&x_l^o=0,&\text{otherwise},\end{array}\]
with $l_j\in\lbra{1,\cdots,n}$ and $\tilde{\theta}_{l_j}$ being the nearest grid to a source $\theta_j$, $j\in\lbra{1,\cdots,k}$. It is easy to show that $\zeron{\m{x}^o}\leq k$, each column of $\m{A}$ and $\m{B}$ has unit norm, and $\m{\beta}^o\in\mbra{-r,r}^n$ with $r=\frac{\kappa}{n}=\frac{\pi}{2n}\sqrt{\frac{m^2-1}{3}}$. In addition, we let $\m{e}=\m{R}\m{s}$ with $\m{R}=\mbra{\m{R}_1,\cdots,\m{R}_k}$, and $\epsilon=\frac{\sqrt{k}\twon{\m{s}}\pi^2}{8n^2}\sqrt{\frac{3m^4-10m^2+7}{15}}$ such that $\twon{\m{e}}\leq\epsilon$ (the information of $k$ and $\twon{\m{s}}$ is used). The derivation for the setting of $\epsilon$ is provided in Appendix F. Then the DOA estimation model can be written into the form of our studied model in (\ref{formu:observationmodel_perturbed}). The only differences are that $\m{A}$, $\m{B}$, $\m{x}^o$ and $\m{e}$ are in the complex domain rather than the real domain and that $\m{e}$ denotes a modeling error term rather than the measurement noise. It is noted that the robust stability results in SP-CS apply straightforward to such complex signal case with few modifications. The objective turns to recovering $\m{x}^o$ (its support actually) and $\m{\beta}^o$. According to Theorem \ref{Thm:BPDN_offgrid} $\m{x}^o$ and $\m{\beta}^o$ can be stably recovered if the D-RIP condition is satisfied. Denote $\widehat{\m{x}}$ the recovered $\m{x}^o$, $\widehat{\m{\beta}}$ the recovered $\m{\beta}^o$, and $\cI$ the support of $\widehat{\m{x}}$. Then we obtain the recovered $\m{\theta}$: $\widehat{\m{\theta}}=\tilde{\m{\theta}}_{\cI}+\kappa^{-1}\widehat{\m{\beta}}_{\cI}$ where $\m{v}_{\cI}$ keeps only entries of a vector $\m{v}$ on the index set $\cI$. The empirical results in Subsection \ref{sec:application} will illustrate the merits of applying the SP-CS framework to estimate DOAs.

\begin{rem}~
 \begin{itemize}
  \item[(1)] Within the scope of DOA estimation, this work is related to spectral CS introduced in \cite{duarte2010spectral}. To obtain an accurate solution, the authors of \cite{duarte2010spectral} adopt a very dense sampling grid (that is necessary for any standard CS based methods according to the mentioned lower bound for the mean squared estimation error) and then prohibit a solution whose support contains near-located indices (that correspond to highly coherent columns in the overcomplete dictionary). In this paper we show that accurate DOA estimation is possible by using a coarse grid and jointly estimating the off-grid distance (the distance from a true DOA to its nearest grid point).
  \item[(2)] The off-grid DOA estimation problem has been studied in \cite{zhu2011sparsity,yang2011off}. The STLS solver in \cite{zhu2011sparsity} obtains a maximum a posteriori solution if $\m{\beta}^o$ is Gaussian distributed. But such a condition does not hold in the off-grid DOA estimation problem. The SBI solver in \cite{yang2011off} proposed by the authors is based on the same model as in (\ref{formu:observationmodel_perturbed}) and within the framework of Bayesian CS \cite{ji2008bayesian}.
  \item[(3)] The proposed P-BPDN can be extended to the multiple measurement vectors case like $\ell_1$-SVD to deal with DOA estimation with multiple snapshots.
 \end{itemize}
\end{rem}

\section{Algorithms for P-BPDN}\label{sec:algorithms}

\subsection{Special Case: Positive Signals}\label{sec:positivesignal}
This subsection studies a
special case where the original signal $\m{x}^o$ is positive-valued
(except zero entries). Such a case has been studied in standard
CS \cite{donoho2009message,stojnic2009various}. By incorporating the
positiveness of $\m{x}^o$, P-BPDN is modified into the positive
P-BPDN (PP-BPDN) problem
\[\min_{\m{x},\m{\beta}}\m{1}^T\m{x},\st \left\{\begin{array}{l}\twon{\m{y}-\sbra{\m{A}+\m{B}\m{\Delta}}\m{x}}\leq\epsilon,\\
\m{x}\succcurlyeq
\m{0},\\r\m{1}\succcurlyeq\m{\beta}\succcurlyeq-r\m{1},\end{array}\right.\]
where $\succcurlyeq$ is $\geq$ with an elementwise operation and
$\m{0}$, $\m{1}$ are column vectors composed of $0$, $1$ respectively
with proper dimensions. It is noted that the robustly stable signal recovery
results in the present paper apply directly to the solution to
PP-BPDN in such case. This subsection shows that the nonconvex
PP-BPDN problem can be transformed into a convex one and hence its
optimal solution can be efficiently obtained. Denote
$\m{p}=\m{\beta}\odot\m{x}$. A new, convex problem $\sbra{P_1}$ is
introduced as follows.
\[\sbra{P_1}\quad\min_{\m{x},\m{p}}\m{1}^T\m{x},\st \left\{\begin{array}{l}\twon{\m{y}-\m{\Psi}\begin{bmatrix}\m{x}\\\m{p}\end{bmatrix}}\leq\epsilon,\\
\m{x}\succcurlyeq \m{0},\\r\m{x}\succcurlyeq\m{p}\succcurlyeq-r\m{x}.\end{array}\right.\]

\begin{thm} Problems PP-BPDN and $\sbra{P_1}$ are equivalent in the sense that, if $\sbra{\m{x}^*,\m{\beta}^*}$ is an optimal solution to PP-BPDN, then there exists $\m{p}^*=\m{\beta}^*\odot\m{x}^*$ such that $\sbra{\m{x}^*,\m{p}^*}$ is an optimal solution to $\sbra{P_1}$, and that, if $\sbra{\m{x}^*,\m{p}^*}$ is an optimal solution to $\sbra{P_1}$, then there exists $\m{\beta}^*$ with $\beta^*_j=\left\{\begin{array}{l}p_j^*/x_j^*,\\0,\end{array} \begin{array}{l}\text{ if }x_j^*>0;\\\text{ otherwise}\end{array}\right.$ such that $\sbra{\m{x}^*,\m{\beta}^*}$ is an optimal solution to PP-BPDN.\label{thm:equivalenceP1P2}
\end{thm}

\begin{proof} We only prove the first part of Theorem \ref{thm:equivalenceP1P2} using contradiction. The second part follows similarly. Suppose that $\sbra{\m{x}^*,\m{p}^*}$ with $\m{p}^*=\m{\beta}^*\odot\m{x}^*$ is not an optimal solution to $\sbra{P_1}$. Then there exists $\sbra{\m{x}',\m{p}'}$ in the feasible domain of $\sbra{P_1}$ such that $\onen{\m{x}'}<\onen{\m{x}^*}$. Define $\m{\beta}'$ as $\beta'_j=\left\{\begin{array}{l}p'_j/x'_j,\\0,\end{array} \begin{array}{l}\text{ if }x'_j>0;\\\text{ otherwise}\end{array}\right.$. It is easy to show that $\sbra{\m{x}',\m{\beta}'}$ is a feasible solution to PP-BPDN. By $\onen{\m{x}'}<\onen{\m{x}^*}$ we conclude that $\sbra{\m{x}^*,\m{\beta}^*}$ is not an optimal solution to PP-BPDN, which leads to contradiction.
\end{proof}

Theorem \ref{thm:equivalenceP1P2} states that an optimal solution to PP-BPDN can be efficiently obtained by solving the convex problem $\sbra{P_1}$.

\subsection{AA-P-BPDN: Alternating Algorithm for P-BPDN}
For general signals, P-BPDN in (\ref{formu:l1_offgrid_Phiv}) is nonconvex. A simple method is to
solve a series of BPDN problems with
{\lentwo\equa{\m{x}^{\sbra{j+1}} &=& \arg\min_{\m{x}}\onen{\m{x}},
\st \notag\\
&&\twon{\m{y}-\sbra{\m{A}+\m{B}\m{\Delta}^{\sbra{j}}}\m{x}}\leq\epsilon,
\label{formu:solve_xj} \\ \m{\beta}^{\sbra{j+1}} &=&
\arg\min_{\m{\beta}\in\mbra{-r,r}^n}
\twon{\m{y}-\sbra{\m{A}+\m{B}\m{\Delta}}\m{x}^{\sbra{j+1}}}\label{formu:solve_betaj}}
}starting from $\m{\beta}^{\sbra{0}}=\m{0}$, where the superscript
$^{\sbra{j}}$ indicates the $j$th iteration and
$\m{\Delta}^{\sbra{j}}=\diag\sbra{\m{\beta}^{\sbra{j}}}$. Denote AA-P-BPDN
the alternating algorithm defined by (\ref{formu:solve_xj}) and
(\ref{formu:solve_betaj}). To analyze AA-P-BPDN, we first present
the following two lemmas.

\begin{lem} For a matrix sequence $\lbra{\m{\Phi}^{\sbra{j}}}_{j=1}^{\infty}$
composed of fat matrices, let $\cD^j=\lbra{\m{v}: \twon{\m{y}-\m{\Phi}^{\sbra{j}}\m{v}}\leq\epsilon}$, $j=1,2,\cdots$, and $\cD^*=\lbra{\m{v}: \twon{\m{y}-\m{\Phi}^*\m{v}}\leq\epsilon}$ with $\epsilon>0$. If $\m{\Phi}^{\sbra{j}}\rightarrow\m{\Phi}^*$, as $j\rightarrow+\infty$, then for any $\m{v}\in\cD^*$ there exists a sequence $\lbra{\m{v}^{\sbra{j}}}_{j=1}^{\infty}$ with $\m{v}^{\sbra{j}}\in\cD^{\sbra{j}}$, $j=1,2,\cdots$, such that $\m{v}^{\sbra{j}}\rightarrow\m{v}$, as $j\rightarrow+\infty$. \label{Lemma:continuousnullspace}
\end{lem}

Lemma \ref{Lemma:continuousnullspace} studies the variation of
feasible domains $\cD^j$, $j=1,2,\cdots$, of a series of BPDN
problems whose sensing matrices $\m{\Phi}^{\sbra{j}}$, $j=1,2,\cdots$,
converge to $\m{\Phi}^*$. It states that the sequence of the feasible
domains also converges to $\cD^*$ in the sense that for any point in
$\cD^*$, there exists a sequence of points, each of which belongs to
one $\cD^j$, that converges to the point. To prove Lemma
\ref{Lemma:continuousnullspace}, we first show that it holds for any
interior point of $\cD^*$ by constructing such a sequence. Then we
show that it also holds for a boundary point of $\cD^*$ by that for
any boundary point there exists a sequence of interior points of
$\cD^*$ that converges to it. The detailed proof is given in
Appendix C.

\begin{lem} An optimal solution $\m{x}^*$ to the BPDN problem in (\ref{formu:BPDN}) satisfies that $\m{x}^*=\m{0}$, if $\twon{\m{y}}\leq\epsilon$, or $\twon{\m{y}-\m{\Phi}\m{x}^*}=\epsilon$, otherwise.\label{Lemma:BPDN_optimalx}
\end{lem}
\begin{proof}
It is trivial for the case where $\twon{\m{y}}\leq\epsilon$.
Consider the other case where $\twon{\m{y}}>\epsilon$. Note first
that $\m{x}^*\neq\m{0}$. We use contradiction to show that the
equality $\twon{\m{y}-\m{\Phi}\m{x}^*}=\epsilon$ holds. Suppose that
$\twon{\m{y}-\m{\Phi}\m{x}^*}<\epsilon$. Introduce
$f\sbra{\theta}=\twon{\m{y}-\theta\m{\Phi}\m{x}^*}$. Then
$f(0)>\epsilon$, and $f(1)<\epsilon$. There exists $\theta_0$,
$0<\theta_0<1$, such that $f\sbra{\theta_0}=\epsilon$ since
$f\sbra{\theta}$ is continuous on the interval $\mbra{0,1}$. Hence,
$\m{x}'=\theta_0\m{x}^*$ is a feasible solution to BPDN in
(\ref{formu:BPDN}). We conclude that $\m{x}^*$ is not optimal by
$\onen{\m{x}'}=\theta_0\onen{\m{x}^*}<\onen{\m{x}^*}$, which leads
to contradiction.
\end{proof}

Lemma \ref{Lemma:BPDN_optimalx} studies the location of an optimal
solution to the BPDN problem. It states that the optimal solution
locates at the origin if the origin is a feasible solution, or at
the boundary of the feasible domain otherwise. This can be easily
observed from Fig. \ref{Fig:illu_solution_BPDN}. Based on Lemmas
\ref{Lemma:continuousnullspace} and \ref{Lemma:BPDN_optimalx}, we
have the following results for AA-P-BPDN.

\begin{thm} Any accumulation point $\sbra{\m{x}^*,\m{\beta}^*}$ of the sequence $\lbra{\sbra{\m{x}^{\sbra{j}},\m{\beta}^{\sbra{j}}}}_{j=1}^{\infty}$ is a stationary point of AA-P-BPDN in the sense that
{\lentwo\equa{\m{x}^*
&=& \arg\min_{\m{x}}\onen{\m{x}}, \st \notag\\ &&\twon{\m{y}-\sbra{\m{A}+\m{B}\m{\Delta}^*}\m{x}}\leq\epsilon, \label{formu:condition_x} \\ \m{\beta}^*
&=& \arg\min_{\m{\beta}\in\mbra{-r,r}^n} \twon{\m{y}-\sbra{\m{A}+\m{B}\m{\Delta}}\m{x}^*}\label{formu:condition_beta}}
}with $\m{\Delta}^*=\diag\sbra{\m{\beta}^*}$. \label{thm:AA-P-BPDN_accumulisstation}
\end{thm}

\begin{thm} An optimal solution $\sbra{\m{x}^*,\m{\beta}^*}$ to P-BPDN in (\ref{formu:l1_offgrid_Phiv}) is a stationary point of AA-P-BPDN. \label{thm:AA-P-BPDN_optimalisstation}
\end{thm}

Theorem \ref{thm:AA-P-BPDN_accumulisstation} studies the
property of the solution $\sbra{\m{x}^{\sbra{j}},\m{\beta}^{\sbra{j}}}$
produced by AA-P-BPDN. It shows that
$\sbra{\m{x}^{\sbra{j}},\m{\beta}^{\sbra{j}}}$ is arbitrarily close to a
stationary point of AA-P-BPDN as the iteration index $j$ is large
enough.\footnote{It is shown in the proof of Theorem
\ref{thm:AA-P-BPDN_accumulisstation} in Appendix D that the
sequence
$\lbra{\sbra{\m{x}^{\sbra{j}},\m{\beta}^{\sbra{j}}}}_{j=1}^{\infty}$ is
bounded. And it can be shown, for example, using contradiction, that
for a bounded sequence $\lbra{a_j}_{j=1}^{\infty}$, there exists an
accumulation point of $\lbra{a_j}_{j=1}^{\infty}$ such that $a_j$ is
arbitrarily close to it as $j$ is large enough.} Hence, the output
of AA-P-BPDN can be considered as a stationary point provided that
an appropriate termination criterion is set. Theorem
\ref{thm:AA-P-BPDN_optimalisstation} tells that an optimal
solution to P-BPDN is a stationary point of AA-P-BPDN. So, it is
possible for AA-P-BPDN to produce an optimal solution to P-BPDN. The
proofs of Theorems \ref{thm:AA-P-BPDN_accumulisstation} and
\ref{thm:AA-P-BPDN_optimalisstation} are provided in Appendix D and
Appendix E respectively.

\begin{rem} During the revision of this paper, we have noted the following formulation of P-BPDN that possibly can provide an efficient approach to an optimal solution to P-BPDN. Let $\m{x}=\m{x}_+-\m{x}_-$ where $\m{x}_+\succcurlyeq\m{0}$, $\m{x}_-\succcurlyeq\m{0}$ and $\m{x}_+\odot\m{x}_-=\m{0}$. Then we have $\abs{\m{x}}=\m{x}_++\m{x}_-$ where $\abs{\cdot}$ applies elementwise. Denote $\m{p}=\m{\beta}\odot\m{x}$. A convex problem can be cast as follows:
\equ{\begin{split}
&\min_{\m{x}_+,\m{x}_-,\m{p}} \m{1}^T\sbra{\m{x}_++\m{x}_-}, \\
&\st \left\{\begin{array}{l} \twon{\m{y}-\begin{bmatrix}\m{A}&-\m{A}& \m{B}\end{bmatrix} \begin{bmatrix}\m{x}_+\\ \m{x}_-\\ \m{p}\end{bmatrix}} \leq\epsilon, \\ \m{x}_+\succcurlyeq\m{0},\\\m{x}_-\succcurlyeq\m{0},\\ r\sbra{\m{x}_++\m{x}_-}\succcurlyeq\m{p}\succcurlyeq-r\sbra{\m{x}_++\m{x}_-}. \end{array}\right.\end{split} \label{formu:convexityofPBPDN}}
\end{rem}
The above convex problem can be considered as a convex relaxation of P-BPDN since it can be shown (like that in Theorem \ref{thm:equivalenceP1P2}) that an optimal solution to P-BPDN can be obtained based on an optimal solution to the problem in (\ref{formu:convexityofPBPDN}) incorporated with an additional nonconvex constraint $\m{x}_+\odot\m{x}_-=\m{0}$. An interesting phenomenon has been observed through numerical simulations that an optimal solution to the problem in (\ref{formu:convexityofPBPDN}) still satisfies the constraint $\m{x}_+\odot\m{x}_-=\m{0}$. Based on such an observation, an efficient approach to P-BPDN is to firstly solve (\ref{formu:convexityofPBPDN}), and then check whether its solution, denoted by $\sbra{\m{x}_+^*,\m{x}_-^*,\m{p}^*}$, satisfies $\m{x}_+\odot\m{x}_-=\m{0}$. If it does, then $\sbra{\m{x}^*,\m{\beta}^*}$ is an optimal solution to P-BPDN where $\m{x}^*=\m{x}_+^*-\m{x}_-^*$ and $\beta^*_j=\left\{\begin{array}{l}p_j^*/x_j^*,\\0,\end{array} \begin{array}{l}\text{ if }x_j^*\neq0;\\\text{ otherwise.}\end{array}\right.$ Otherwise, we may turn to AA-P-BPDN again. But we note that it is still an open problem whether an optimal solution to (\ref{formu:convexityofPBPDN}) always satisfies the constraint $\m{x}_+\odot\m{x}_-=\m{0}$. In addition, the convex relaxation in (\ref{formu:convexityofPBPDN}) does not apply to the complex signal case as in DOA estimation studied in Subsection \ref{sec:DOAapplication}.

\subsection{Effectiveness of AA-P-BPDN}\label{sec:effectiv_AA-P-BPDN}

As reported in the last subsection, it is possible for AA-P-BPDN to
produce an optimal solution to P-BPDN. But it is not easy to check
the optimality of the output of AA-P-BPDN because of the
nonconvexity of P-BPDN. Instead, we study the effectiveness of
AA-P-BPDN in solving P-BPDN in this subsection with the concept of
effectiveness as defined in Subsection
\ref{sec:Concept_Effectiveness}. By Corollary
\ref{cor:optimalnotnecessary_offgrid}, a good signal recovery
$\widehat{\m{x}}$ of $\m{x}^o$ is not necessarily an optimal solution.
It requires only that $\sbra{\widehat{\m{x}},\widehat{\m{\beta}}}$, where
$\widehat{\m{\beta}}$ denotes the recovery of $\m{\beta}^o$, be a feasible
solution to P-BPDN and that $\onen{\widehat{\m{x}}}\leq\onen{\m{x}^o}$
holds. As shown in the proof of Theorem
\ref{thm:AA-P-BPDN_accumulisstation} in Appendix D, that
$\sbra{\m{x}^{\sbra{j}},\m{\beta}^{\sbra{j}}}$ for any $j\geq1$ is a
feasible solution to P-BPDN and that the sequence
$\lbra{\onen{\m{x}^{\sbra{j}}}}_{j=1}^{\infty}$ is monotone
decreasing and converges. So, the effectiveness of AA-P-BPDN in
solving P-BPDN can be assessed via numerical simulations by checking
whether $\onen{\m{x}^{AA}}\leq\onen{\m{x}^o}$ holds with
$\m{x}^{AA}$ denoting the output of AA-P-BPDN. The effectiveness of
AA-P-BPDN is verified in Subsection \ref{sec:academicsimulation} via numerical simulations, where we observe
that the inequality $\onen{\m{x}^{AA}}\leq\onen{\m{x}^o}$ holds in
all experiments (over $3700$ trials).

\section{Numerical Simulations}\label{sec:simulation}

\subsection{Verification of the Robust Stability}\label{sec:academicsimulation}
This subsection demonstrates the robustly stable signal recovery results of
SP-CS in the present paper, as well as the effectiveness of AA-P-BPDN in solving P-BPDN in (\ref{formu:l1_offgrid_Phiv}), via numerical simulations. AA-P-BPDN is implemented in
Matlab with problems in (\ref{formu:solve_xj}) and
(\ref{formu:solve_betaj}) being solved using CVX
\cite{grant2008cvx}. AA-P-BPDN is terminated as
$\frac{\abs{\onen{\m{x}^{\sbra{j}}}-\onen{\m{x}^{\sbra{j-1}}}}}
{\onen{\m{x}^{\sbra{j-1}}}}\leq1\times10^{-6}$ or the maximum number
of iterations, set to $200$, is reached. PP-BPDN is also implemented in Matlab and solved by CVX.

We first consider general signals. The
sparse signal case is mainly studied. The variation of the signal recovery error is studied with respect to the noise level, perturbation level and number of measurements respectively. Besides AA-P-BPDN for P-BPDN in SP-CS, performances of three other approaches are also studied. The first one assumes that the perturbation is known {\em a priori} and recovers the original signal $\m{x}^o$ by solving, namely, the oracle (O-) BPDN problem
\equ{\min_{\m{x}}\onen{\m{x}},\st\twon{\m{y}-\sbra{\m{A}+\m{B}\m{\Delta}^o}\m{x}}\leq\epsilon.\notag}
The O-BPDN approach produces the best recovery result of SP-CS within the scope of $\ell_1$ minimization of CS since it exploits the exact perturbation (oracle information). The second one corresponds to the robust signal recovery of perturbed CS as described in Subsection \ref{sec:robustrecovery} and solves N-BPDN in (\ref{formu:N-BPDN}) where $\epsilon_{\m{E},\m{x}^o}=\twon{\m{B}\m{\Delta}^o\m{x}^o}$ is used though it is not available in practice. The last one refers to the other approach to SP-CS that seeks for the signal recovery by solving TPS-BPDN in (\ref{formu:TPS-BPDN}) as discussed in Subsection \ref{sec:Concept_Effectiveness}.

The first experiment studies the signal recovery error with respect to the noise level. We set the signal length $n=200$, sample size $m=80$, sparsity level $k=10$ and perturbation parameter $r=0.1$. The noise level $\epsilon$ varies from $0.05$ to $2$ with interval $0.05$. For each combination of $\sbra{n,m,k,r,\epsilon}$, the signal recovery error, as well as $\m{\beta}^o$ recovery error (on the support of $\m{x}^o$), is averaged over $R=50$ trials. In each trial, matrices $\m{A}$ and $\m{B}$ are generated from Gaussian distribution and each column of them has zero mean and unit norm after proper scaling. The sparse signal $\m{x}^o$ is composed of unit spikes with random signs and locations. Entries of $\m{\beta}^o$ are uniformly distributed in $\mbra{-r,r}$. The noise $\m{e}$ is zero mean Gaussian distributed and then scaled such that $\twon{\m{e}}=\epsilon$. Using the same data, the four approaches, including O-BPDN, N-BPDN, TPS-BPDN and AA-P-BPDN for P-BPDN, are used to recover $\m{x}^o$ respectively in each trial. The simulation results are shown in Fig. \ref{Fig:Error_Noise}. It can be seen that both signal and $\m{\beta}^o$ recovery errors of AA-P-BPDN for P-BPDN in SP-CS are proportional to the noise, which is consistent with our robustly stable signal recovery result in the present paper. The error of N-BPDN grows linearly with the noise but a large error still exhibits in the noise free case. Except the ideal case of O-BPDN, our proposed P-BPDN has the smallest error.

\begin{figure}
  \centering
  \includegraphics[width=3.5in]{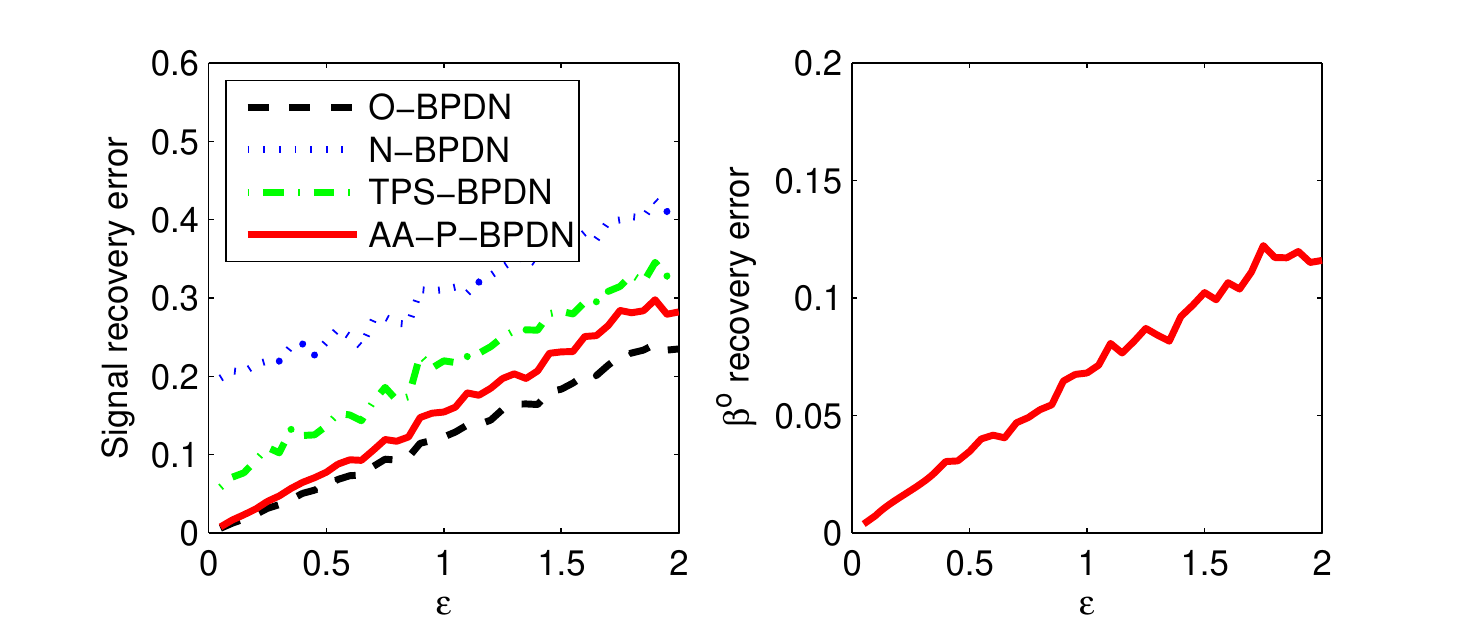}
  \caption{Signal and perturbation recovery errors with respect to the noise level $\epsilon$ with parameter settings $\sbra{n,m,k,r}=\sbra{200,80,10,0.1}$. Both signal and $\m{\beta}^o$ recovery errors of AA-P-BPDN for P-BPDN in SP-CS are proportional to $\epsilon$.}
  \label{Fig:Error_Noise}
\end{figure}

The second experiment studies the effect of the structured perturbation. Experiment settings are the same as those in the first experiment except that we set $\sbra{n,m,k,\epsilon}=\sbra{200,80,10,0.5}$ and vary $r\in\lbra{0.05,0.1,\cdots,1}$. Fig. \ref{Fig:Error_Perturbation} presents our simulation results. A nearly constant error is obtained using O-BPDN in standard CS since the perturbation is assumed to be known in O-BPDN. The error of AA-P-BPDN for P-BPDN in SP-CS slowly increases with the perturbation level and is quite close to that of O-BPDN for a moderate perturbation. Such a behavior is consistent with our analysis. Besides, it can be observed that the error of N-BPDN grows linearly with the perturbation level. Again, our proposed P-BPDN has the smallest error except O-BPDN.

\begin{figure}
  \centering
  \includegraphics[width=3.5in]{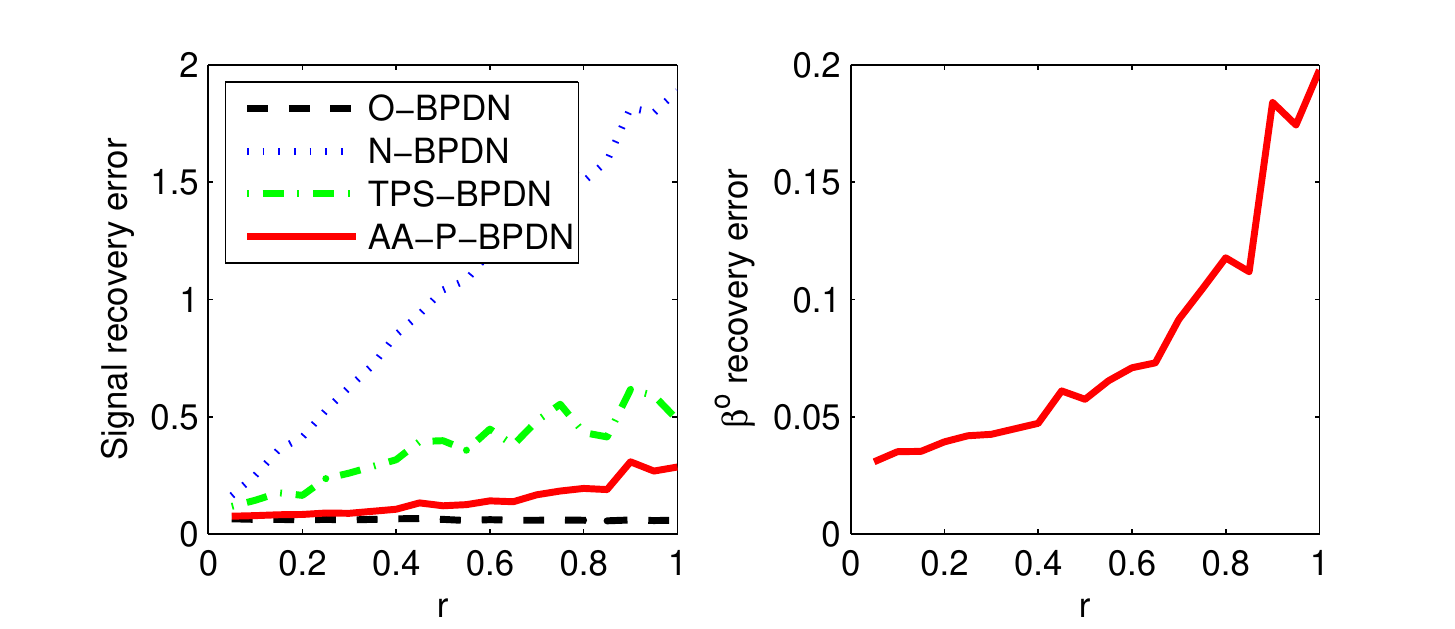}
  \caption{Signal and perturbation recovery errors with respect to the perturbation level in terms of $r$ with parameter settings $\sbra{n,m,k,\epsilon}=\sbra{200,80,10,0.5}$. The error of AA-P-BPDN for P-BPDN in SP-CS slowly increases with the perturbation level and is quite close to that of the ideal case of O-BPDN for a moderate perturbation.}
  \label{Fig:Error_Perturbation}
\end{figure}

The third experiment studies the variation of the recovery error with the number of measurements. We set $\sbra{n,k,r,\epsilon}=\sbra{200,10,0.1,0.2}$ and vary $m\in\lbra{30,35,\cdots,100}$. Simulation results are presented in Fig. \ref{Fig:Error_MeasureNumber}. Signal recovery errors of all four approaches decrease as the number of measurements increases. Again, it is observed that O-BPDN of the ideal case achieves the best result followed by our proposed P-BPDN. For example, to obtain the signal recovery error of $0.05$, about $55$ measurements are needed for O-BPDN while the numbers are, respectively, $65$ for AA-P-BPDN and $95$ for TPS-BPDN. It is impossible for N-BPDN to achieve such a small error in our observation because of the existence of the perturbation.

\begin{figure}
  \centering
  \includegraphics[width=3.5in]{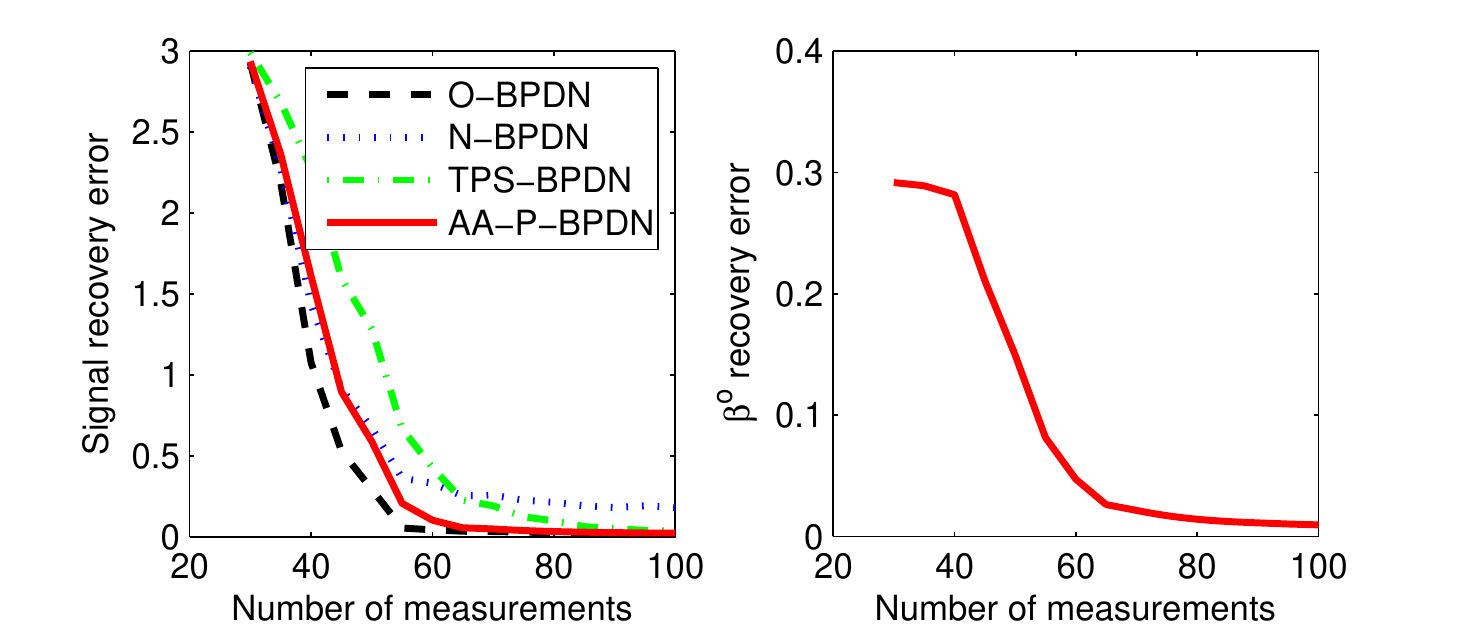}
  \caption{Signal and perturbation recovery errors with respect to the number of measurements with parameter settings $\sbra{n,k,r,\epsilon}=\sbra{200,10,0.1,0.2}$. AA-P-BPDN for P-BPDN in SP-CS has the best performance except the ideal case of O-BPDN.}
  \label{Fig:Error_MeasureNumber}
\end{figure}

We next consider a compressible signal that is generated by taking a fixed sequence $\lbra{2.8843\cdot j^{-1.5}}_{j=1}^n$ with $n=200$, randomly permuting it, and multiplying by a random sign sequence (the coefficient $2.8843$ is chosen such that the compressible signal has the same $\ell_2$ norm as the sparse signals in the previous experiments). It is sought to be recovered from $m=70$ noisy measurements with $\epsilon=0.2$ and $r=0.1$. Give experiment results in one instance as an example. The signal recovery error of AA-P-BPDN for P-BPDN in SP-CS is about $0.239$, while errors of O-BPDN, N-BPDN and TPS-BPDN are about $0.234$, $0.361$ and $0.314$ respectively.


For the special positive signal case, an optimal solution to PP-BPDN can be efficiently obtained. An experiment result is shown in Fig. \ref{Fig:PositiveSparseSignal}, where a sparse signal of length $n=200$, composed of $k=10$ positive unit spikes, is exactly recovered from $m=50$ noise free measurements with $r=0.1$ by solving $\sbra{P_1}$.

\begin{figure}
  \centering
  \includegraphics[width=3.5in]{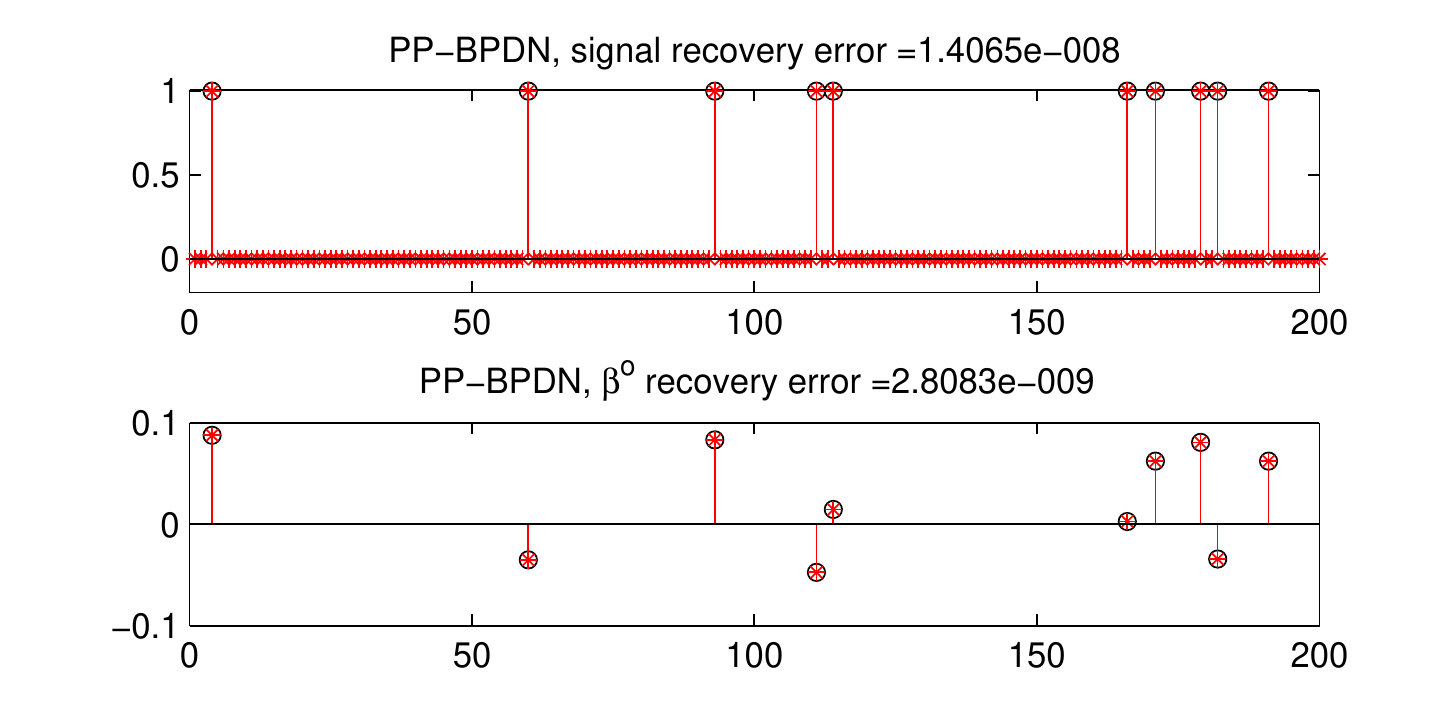}
  \caption{Exact recovery of a positive sparse signal from noise-free measurements with $\sbra{m,n,k,r,\epsilon}=\sbra{200,50,10,0.1,0}$. PP-BPDN is solved by solving $\sbra{P_1}$. $\m{\beta}^o$ and its recovery are shown only on the support of $\m{x}^o$. Black circles: original signal and $\m{\beta}^o$; red stars: recoveries.}
  \label{Fig:PositiveSparseSignal}
\end{figure}

\subsection{Empirical Results of DOA Estimation} \label{sec:application}
This subsection studies the empirical performance of the application of the studied SP-CS framework in DOA estimation. We consider the case of $n=90$ and $k=2$. Numerical calculations show that the D-RIP condition $\bar{\delta}_{4k}\sbra{\m{\Psi}}<\sbra{\sqrt{2\sbra{1+r^2}}+1}^{-1}$ in Theorem \ref{Thm:BPDN_offgrid} is satisfied if $m\geq145$. Though it ceases to be a ``compressed'' sensing problem in the case $m\geq n$, it still makes sense in SP-CS since there are $2n$ variables to be estimated and hence the P-BPDN problem is still underdetermined as $m < 2n$. As noted in Subsection \ref{sec:discussion}, the D-RIP condition can be possibly relaxed using recent techniques in standard CS, which may reduce the required $m$ value. In addition, a RIP condition is a sufficient condition for guaranteed signal recovery accuracy while its conservativeness in standard CS has been studied in \cite{blanchard2011compressed}. We next choose a much smaller $m=30$ ($r\approx0.302$ in such a case) and show the empirical performance of the proposed SP-CS framework on such off-grid DOA estimation.


The experimental setup is as follows. In each trial, the complex source signal $\m{s}$ is generated with both entries having unit amplitude and random phases. $\theta_1$ and $\theta_2$ are generated uniformly from intervals $\mbra{\frac{2}{n},\frac{4}{n}}$ and $\mbra{\frac{12}{n},\frac{14}{n}}$ respectively ($5.1^\circ\sim7.7^\circ$ apart in the DOA domain). P-BPDN is solved using AA-P-BPDN whose settings are the same as those in Subsection \ref{sec:academicsimulation}. Our experimental results of the estimation error $\widehat{\m{\theta}}-\m{\theta}$ for both sources are presented in Fig. \ref{Fig:DOAErrorDistribution} where 1000 trials are used. It can be seen that P-BPDN performs well on the off-grid DOA estimation. All estimation errors lie in the interval $\mbra{-\frac{1}{n},\frac{1}{n}}$ with most very close to zero. To achieve a possibly comparable mean squared estimation error, a grid of length at least $n=360$ has to be used in standard CS based methods according to the lower bound mentioned in Subsection \ref{sec:DOAapplication}. An example of performance of SP-CS and standard CS on DOA estimation is shown in Fig. \ref{Fig:SPCSvsStdCS}, where the two approaches share the same data set and $n=360$ is set in standard CS. From the upper two sub-figures, it can be seen that SP-CS performs well on both source signal and $\m{\beta}^o$ recoveries. From the lower left one, however, it can be seen that two nonzero entries are presented in the recovered signal around the location of each source when using standard CS. Such a phenomenon is much clearer in the last sub-figure, where it can be observed that a single peak exhibits at a place very close to the true location of source 1 using the proposed SP-CS framework while two peaks occurs at places further away from the true source in standard CS.

\begin{figure}
  \centering
  \includegraphics[width=3.5in]{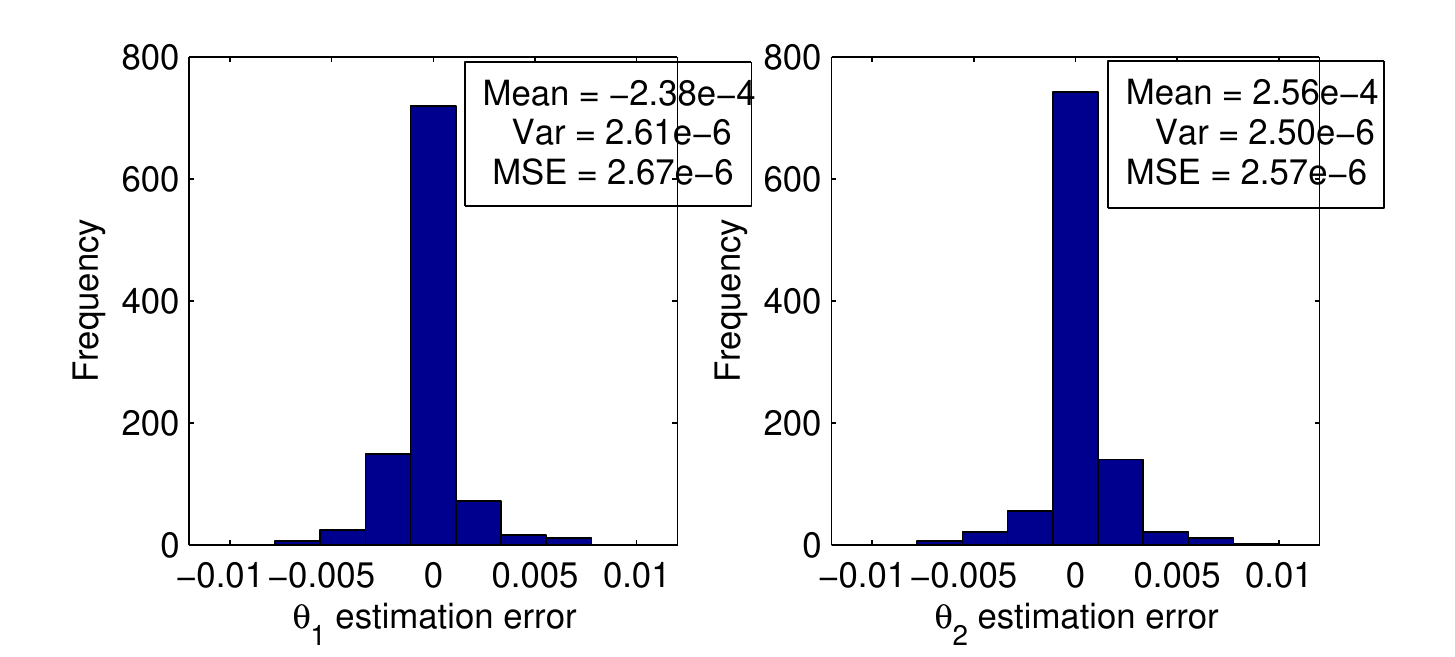}
  \caption{Histogram of $\m{\theta}$ estimation error for both sources using P-BPDN for SP-CS. Statistics including mean, variance and mean squared error (MSE) are shown.}
  \label{Fig:DOAErrorDistribution}
\end{figure}

\begin{figure}
  \centering
  \includegraphics[width=3.5in]{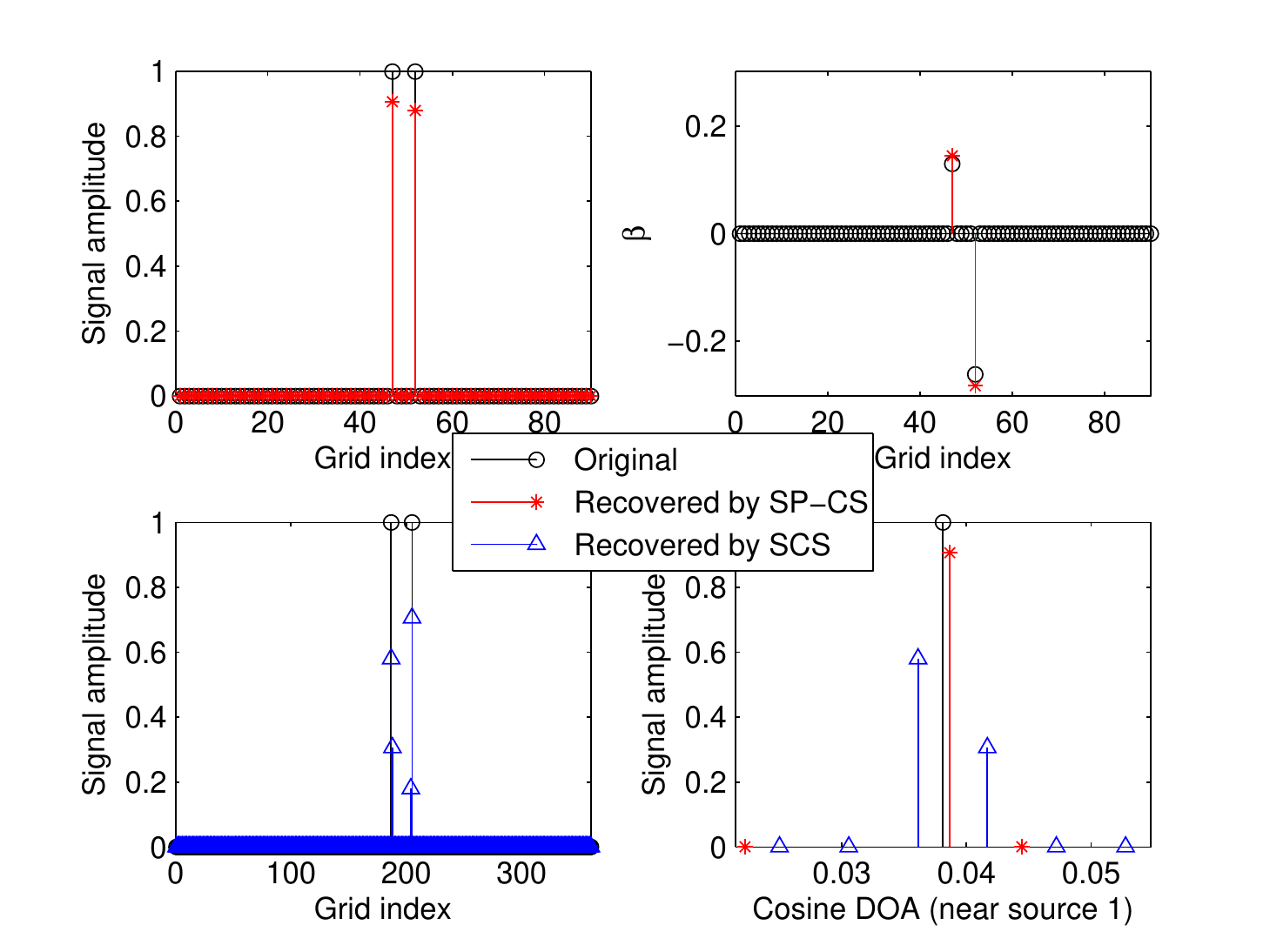}
  \caption{Performance comparison of SP-CS and standard CS (SCS) on DOA estimation. Upper left: signal recovery in SP-CS; upper right: $\m{\beta}^o$ recovery in SP-CS (shown only on the signal support); lower left: signal recovery in standard CS; lower right: signal amplitude versus $\theta$ (near the location of source 1) in SP-CS and standard CS.}
  \label{Fig:SPCSvsStdCS}
\end{figure}

%

\section{Conclusion}\label{sec:conclusion}
This paper studied the CS problem in the presence of measurement noise and a structured matrix perturbation. A concept named as robust stability for signal recovery was introduced. It was shown that the robust stability can be achieved for a sparse signal by solving an $\ell_1$ minimization problem P-BPDN under mild conditions. In the presence of measurement noise, the recovery error is at most proportional to the noise level and the recovery is exact in the special noise free case. A general result for compressible signals was also reported. An alternating algorithm named as AA-P-BPDN was proposed to solve the nonconvex P-BPDN problem, and numerical simulations were carried out, verifying our theoretical analysis. A practical application in DOA estimation was studied and satisfactory estimation results were obtained.

The simulation results of DOA estimation suggest that the RIP condition for the robust stability is quite conservative in practice. One future work is to relax such a condition. In our problem formulation, the signal $\m{x}^o$ and $\m{\beta}^o$ that determines the matrix perturbation are jointly sparse. While this paper focuses on extracting the information that $\m{x}^o$ is sparse and that each entry of $\m{\beta}^o$ lies in a bounded interval, such joint sparsity is not exploited. Inspired by the recent works on block and structured sparsity, e.g., \cite{eldar2010block,bach2011structured}, one future direction is to take into account the joint sparsity information in the signal recovery process to obtain possibly improved recovery performance. Our studied perturbed CS problem is related to the area of dictionary learning for sparse representation \cite{aharon2006ksvd}, where there is typically no {\em a priori} known structure in the overcomplete dictionary and a large number of observation vectors are important to make the learning process succeed. The studied problem in this paper can be considered as a dictionary learning problem but with a known structure in the dictionary, which leads to some similarity between our optimization approach and algorithms for dictionary learning, e.g., $K$-SVD \cite{aharon2006ksvd} and MOD \cite{engan2000multi}. Due to the known structure, it has been shown in this paper that a single observation vector is enough to learn the dictionary with guaranteed performance. Further relations deserve future studies.


\section*{Appendix A\\Proof of Theorem \ref{Thm:l0_offgrid}}
Denote $\m{z}=\begin{bmatrix}\m{x}\\\m{\beta}\odot\m{x}\end{bmatrix}$ and similarly define $\m{z}^o$ and $\m{z}^*$. Then the problem in (\ref{formu:l0_offgrid_Phiv}) can be rewritten into
\equ{\min_{\m{x}\in\bR^n,\m{\beta}\in\mbra{-r,r}^n}\zeron{\m{x}}, \st \m{y}=\m{\Psi}\m{z}. \label{formu:l0_offgrid_Psiz}}
Let $\bar\delta_{k}=\bar\delta_{k}\sbra{\m{\Psi}}$ hereafter for brevity.

First note that $\m{x}^o$ is $k$-sparse and $\m{z}^o$ is $2k$-D-sparse. Since $\sbra{\m{x}^*,\m{\beta}^*}$ is a solution to the problem in (\ref{formu:l0_offgrid_Psiz}), we have $\zeron{\m{x}^*}\leq\zeron{\m{x}^o}\leq k$ and, hence, $\m{z}^*$ is $2k$-D-sparse. By $\m{y}=\m{\Psi}\m{z}^o=\m{\Psi}\m{z}^*$ we obtain $\m{\Psi}\sbra{\m{z}^o-\m{z}^*}=\m{0}$ and thus $\m{z}^o-\m{z}^*=\m{0}$ by $\bar{\delta}_{4k}<1$ and the fact that $\m{z}^o-\m{z}^*$ is $4k$-D-sparse. We complete the proof by observing that $\m{z}^o-\m{z}^*=\begin{bmatrix} \m{x}^o-\m{x}^*\\\m{\beta}^o\odot\m{x}^o-\m{\beta}^*\odot\m{x}^*\end{bmatrix}=\m{0}$.

\section*{Appendix B\\Proofs of Theorems \ref{Thm:BPDN_offgrid} and \ref{Thm:BPDN_offgrid_compressible}}

We only present the proof of Theorem \ref{Thm:BPDN_offgrid_compressible} since Theorem \ref{Thm:BPDN_offgrid} is a special case of Theorem \ref{Thm:BPDN_offgrid_compressible}. We first show the following lemma.

\begin{lem} We have
\equ{\abs{\inp{\m{\Psi}\m{v},\m{\Psi}\m{v}'}}\leq\bar{\delta}_{2(k+k')}\twon{\m{v}}\twon{\m{v}'}\notag}
for all $2k$-D-sparse $\m{v}$ and $2k'$-D-sparse $\m{v}'$ supported on disjoint subsets.\label{Lamma:RIP_disjointsparse}
\end{lem}

\begin{proof} Without loss of generality, assume that $\m{v}$ and $\m{v}'$ are unit vectors with disjoint supports as above. Then by the definition of D-RIP and $\twon{\m{v}\pm\m{v}'}^2=\twon{\m{v}}^2 +\twon{\m{v}'}^2=2$ we have
\equ{2\sbra{1-\bar{\delta}_{2\sbra{k+k'}}}\leq\twon{\m{\Psi}\m{v}\pm\m{\Psi}\m{v}'}^2\leq 2\sbra{1+\bar{\delta}_{2\sbra{k+k'}}}.\notag}
And thus
\equ{\abs{\inp{\m{\Psi}\m{v},\m{\Psi}\m{v}'}}\leq\frac{1}{4}\abs{\twon{\m{\Psi}\m{v}+\m{\Psi}\m{v}'}^2 -\twon{\m{\Psi}\m{v}-\m{\Psi}\m{v}'}^2}\leq\bar\delta_{2\sbra{k+k'}},\notag}
which completes the proof.
\end{proof}

Using the notations $\m{z}$, $\m{z}^o$, $\m{z}^*$ and $\bar\delta_{k}$ in Appendix A, P-BPDN in (\ref{formu:l1_offgrid_Phiv}) can be rewritten into
\equ{\min_{\m{x}\in\bR^n,\m{\beta}\in\mbra{-r,r}^n}\onen{\m{x}}, \st \twon{\m{y}-\m{\Psi}\m{z}}\leq\epsilon. \label{formu:l1_offgrid_Psiz}}
Let $\m{h}=\m{x}^*-\m{x}^o$ and decompose $\m{h}$ into a sum of $k$-sparse vectors $\m{h}_{T_0},\m{h}_{T_1},\m{h}_{T_2},\cdots$, where $T_0$ denotes the set of indices of the $k$ largest entries (in absolute value) of $\m{x}^o$, $T_1$ the set of the $k$ largest entries of $\m{h}_{T_0^c}$ with $T_0^c$ being the complementary set of $T_0$, $T_2$ the set of the next $k$ largest entries of $\m{h}_{T_0^c}$ and so on. We abuse notations $\m{z}_{T_j}^*=\begin{bmatrix}\m{x}_{T_j}^*\\ \m{\beta}_{T_j}^*\odot\m{x}_{T_j}^*\end{bmatrix}$, $j=0,1,2,\cdots$, and similarly define $\m{z}_{T_j}^o$. Let $\m{f}=\m{z}^*-\m{z}^o$ and $\m{f}_{T_j}=\m{z}_{T_j}^*-\m{z}_{T_j}^o$ for $j=0,1,2,\cdots$. For brevity we write $T_{01}=T_0\cup T_1$. To bound $\twon{\m{h}}$, in the first step we show that $\twon{\m{h}_{T_{01}^c}}$ is essentially bounded by $\twon{\m{h}_{T_{01}}}$, and then in the second step we show that $\twon{\m{h}_{T_{01}}}$ is sufficiently small.

The first step follows from the proof of Theorem 1.3 in \cite{candes2008restricted}. Note that
\equ{\twon{\m{h}_{T_j}}\leq k^{1/2}\inftyn{\m{h}_{T_j}}\leq k^{-1/2}\onen{\m{h}_{T_{j-1}}}, \quad j\geq2,}
and thus
\equ{\begin{split}\twon{\m{h}_{T_{01}^c}}
&=\twon{\sum_{j\geq2}\m{h}_{T_j}} \leq\sum_{j\geq2}\twon{\m{h}_{T_j}}\\
&\leq k^{-1/2}\sum_{j\geq1}\onen{\m{h}_{T_j}}\leq k^{-1/2}\onen{\m{h}_{T_0^c}}.\end{split}\label{formu:hT01c}}
Since $\m{x}^*=\m{x}^o+\m{h}$ is an optimal solution, we have
\equ{\begin{split}\onen{\m{x}^o}
&\geq\onen{\m{x}^o+\m{h}}=\sum_{j\in T_0}\abs{x_j^o+h_j}+\sum_{j\in T_0^c}\abs{x_j^o+h_j}\\
&\geq\onen{\m{x}_{T_0}^o}-\onen{\m{h}_{T_0}}+\onen{\m{h}_{T_0^c}} -\onen{\m{x}_{T_0^c}^o}\end{split}}
and thus
\equ{\onen{\m{h}_{T_0^c}}\leq\onen{\m{h}_{T_0}}+2\onen{\m{x}_{T_0^c}^o}.\label{formu:hT0c}}
By (\ref{formu:hT01c}), (\ref{formu:hT0c}) and the inequality $\onen{\m{h}_{T_0}}\leq k^{1/2}\twon{\m{h}_{T_0}}$ we have
\equ{\twon{\m{h}_{T_{01}^c}}\leq\sum_{j\geq2}\twon{\m{h}_{T_j}}\leq\twon{\m{h}_{T_0}}+2k^{-1/2}e_0 \label{formu:formu4}}
with $e_0\equiv\onen{\m{x}^o-\m{x}^k}$.

In the second step, we bound $\twon{\m{h}_{T_{01}}}$ by utilizing its relationship with $\twon{\m{f}_{T_{01}}}$. Note that $\m{f}_{T_j}$ for each $j=0,1,\cdots$ is $2k$-D-sparse. By $\m{\Psi}\m{f}_{T_{01}}=\m{\Psi}\m{f}-\sum_{j\geq2}\m{\Psi}\m{f}_{T_j}$ we have
{\lentwo\equa{
&&\twon{\m{\Psi}\m{f}_{T_{01}}}^2= \inp{\m{\Psi}\m{f}_{T_{01}},\m{\Psi}\m{f}}-\sum_{j\geq2}\inp{\m{\Psi}\m{f}_{T_{01}},\m{\Psi}\m{f}_{T_j}}\notag\\
&\leq& \abs{\inp{\m{\Psi}\m{f}_{T_{01}},\m{\Psi}\m{f}}} + \sum_{j\geq2}\abs{\inp{\m{\Psi}\m{f}_{T_0},\m{\Psi}\m{f}_{T_j}}} \notag\\
&& +\sum_{j\geq2}\abs{\inp{\m{\Psi}\m{f}_{T_1},\m{\Psi}\m{f}_{T_j}}}\notag\\
&\leq& \twon{\m{\Psi}\m{f}_{T_{01}}}\cdot\twon{\m{\Psi}\m{f}} + \bar{\delta}_{4k}\twon{\m{f}_{T_0}} \sum_{j\geq2}\twon{\m{f}_{T_j}}\notag\\
&&+ \bar{\delta}_{4k}\twon{\m{f}_{T_1}} \sum_{j\geq2}\twon{\m{f}_{T_j}}\label{formu:formu1}\\
&\leq& \twon{\m{f}_{T_{01}}} \sbra{2\epsilon\sqrt{1+\bar{\delta}_{4k}} + \sqrt{2}\bar{\delta}_{4k} \sum_{j\geq2}\twon{\m{f}_{T_j}}}.\label{formu:formu2}}
}We used Lemma \ref{Lamma:RIP_disjointsparse} in (\ref{formu:formu1}). In (\ref{formu:formu2}), we used the D-RIP, and inequalities $\twon{\m{f}_{T_0}}+\twon{\m{f}_{T_1}}\leq\sqrt{2}\twon{\m{f}_{T_{01}}}$ and
\equ{\twon{\m{\Psi}\m{f}}=\twon{\m{\Psi}\sbra{\m{z}^*-\m{z}^o}}\leq\twon{\m{y}-\m{\Psi}\m{z}^*}+ \twon{\m{y}-\m{\Psi}\m{z}^o}\leq2\epsilon.\label{formu:Psifbound}}
By noting that $\m{\beta}^o,\m{\beta}^*\in\mbra{-r,r}^n$ and \equ{\m{f}=\begin{bmatrix}\m{h}\\\m{\beta}^*\odot\m{h}+\sbra{\m{\beta}^*-\m{\beta}^o}\odot\m{x}^o\end{bmatrix} \label{formu:f}}
we have
\equ{\twon{\m{f}_{T_j}}\leq\sqrt{1+r^2}\twon{\m{h}_{T_j}}+2r\twon{\m{x}_{T_j}^o}, \quad j=0,1,\cdots.\label{formu:formu3}}
Meanwhile,
\equ{\sum_{j\geq2}\twon{\m{x}_{T_j}^o}\leq\sum_{j\geq2}\onen{\m{x}_{T_j}^o}= \onen{\m{x}_{T_{01}^c}^o}\leq e_0.\label{formu:formu5}}
Applying the D-RIP, (\ref{formu:formu2}), (\ref{formu:formu3}) and then (\ref{formu:formu4}) and (\ref{formu:formu5}) it gives
\equ{\begin{split}
&\sbra{1-\bar{\delta}_{4k}}\twon{\m{f}_{T_{01}}}^2\leq\twon{\m{\Psi}\m{f}_{T_{01}}}^2 \\ \leq&\twon{\m{f}_{T_{01}}}\Big\{2\sqrt{1+\bar{\delta}_{4k}}\epsilon +\sqrt{2\sbra{1+r^2}}\bar{\delta}_{4k}\twon{\m{h}_{T_0}}\\ &\qquad\qquad+2\sqrt{2}\bar{\delta}_{4k}\mbra{\sqrt{1+r^2}k^{-1/2}+r}e_0 \Big\},\end{split}\notag}
and thus
\equ{\begin{split}
&\twon{\m{h}_{T_{01}}}\leq\twon{\m{f}_{T_{01}}}\\
\leq& c_1\epsilon +c_0\twon{\m{h}_{T_{01}}} +\sbra{c_2k^{-1/2}+c_3}e_0\end{split}\notag}
with $c_0\equiv\frac{\sqrt{2\sbra{1+r^2}}\bar{\delta}_{4k}}{1-\bar{\delta}_{4k}}$, $c_1\equiv\frac{2\sqrt{1+\bar{\delta}_{4k}}}{1-\bar{\delta}_{4k}}$, $c_2\equiv2c_0$ and $c_3\equiv\frac{2\sqrt{2}\bar{\delta}_{4k}r}{1-\bar{\delta}_{4k}}$. Hence, we get a bound
\equ{\twon{\m{h}_{T_{01}}}\leq\sbra{1-c_0}^{-1}\mbra{c_1\epsilon+\sbra{c_2k^{-1/2}+c_3}e_0},\notag}
which together with (\ref{formu:formu4}) gives
\equ{\begin{split}
&\twon{\m{h}}\leq\twon{\m{h}_{T_{01}}}+\twon{\m{h}_{T_{01}^c}} \\ \leq&2\twon{\m{h}_{T_{01}}}+2k^{-1/2}e_0\\
\leq& \frac{2c_1}{1-c_0}\epsilon +\mbra{\sbra{\frac{2c_2}{1-c_0}+2}k^{-1/2} +\frac{2c_3}{1-c_0}}e_0,\end{split}\label{formu:boundh}}
which concludes (\ref{formu:boundx}).

By (\ref{formu:Psifbound}), (\ref{formu:f}), (\ref{formu:boundh}) and the RIP we have
\[\begin{split}
&\mbra{1-\delta_k\sbra{\m{B}}}^{1/2}\twon{\sbra{\m{\beta}_{T_0}^*-\m{\beta}^o_{T_0}}\odot \m{x}_{T_0}^o}\\
\leq&\twon{\m{B}\mbra{\sbra{\m{\beta}_{T_0}^*-\m{\beta}^o_{T_0}}\odot\m{x}_{T_0}^o}}\\
=&\twon{\m{\Psi}\begin{bmatrix}\m{0}\\\sbra{\m{\beta}_{T_0}^*-\m{\beta}^o_{T_0}}\odot\m{x}_{T_0}^o\end{bmatrix}} \\ =&\twon{\m{\Psi}\sbra{\m{f}-\begin{bmatrix}\m{h}\\\m{\beta}^*\odot\m{h}\end{bmatrix}- \begin{bmatrix}\m{0}\\\sbra{\m{\beta}_{T_0^c}^*-\m{\beta}^o_{T_0^c}}\odot\m{x}_{T_0^c}^o\end{bmatrix}}}\\
\leq& \twon{\m{\Psi}\m{f}}+\twon{\m{\Psi}\begin{bmatrix}\m{h}\\\m{\beta}^*\odot\m{h}\end{bmatrix}} +\twon{\m{B}\mbra{\sbra{\m{\beta}_{T_0^c}^*-\m{\beta}^o_{T_0^c}}\odot\m{x}_{T_0^c}^o}}\\
\leq& 2\epsilon + \sqrt{1+r^2}\twon{\m{\Psi}}\twon{\m{h}} + 2r\twon{\m{B}}e_0\\
\leq& c_4\epsilon+\sbra{c_5k^{-1/2}+c_6}e_0
\end{split}\notag\]
with $\delta_k\sbra{\m{B}}\leq\delta_{2k}\sbra{\m{B}}\leq\bar{\delta}_{4k}$,  $c_4\equiv2+\frac{2\sqrt{1+r^2}\twon{\m{\Psi}}c_1}{1-c_0}$, $c_5\equiv \sqrt{1+r^2}\sbra{\frac{2c_2}{1-c_0}+2}\twon{\m{\Psi}}$ and $c_6\equiv \sbra{\frac{2\sqrt{1+r^2}c_3}{1-c_0}+2r}\twon{\m{\Psi}}$, and thus
\equ{\begin{split}
&\twon{\sbra{\m{\beta}_{T_0}^*-\m{\beta}^o_{T_0}}\odot \m{x}_{T_0}^o}\\
\leq& \frac{1}{\sqrt{1-\bar\delta_{4k}}}\mbra{c_4\epsilon+\sbra{c_5k^{-1/2}+c_6}e_0}, \end{split}\notag}
which concludes (\ref{formu:boundbetax}). We complete the proof by noting that the above results make sense if $c_0<1$, i.e.,
\equ{\bar{\delta}_{4k}<\frac{1}{\sqrt{2\sbra{1+r^2}}+1}.\notag}


\section*{Appendix C\\Proof of Lemma \ref{Lemma:continuousnullspace}}
We first consider the case where $\m{v}$ is an interior point of $\cD^*$, i.e., it holds that $\twon{\m{y}-\m{\Phi}^*\m{v}}=\epsilon_0<\epsilon$. Let $\eta=\epsilon-\epsilon_0$. Construct a sequence $\lbra{\m{v}^{\sbra{j}}}_{j=1}^{\infty}$ such that $\twon{\m{v}^{\sbra{j}}-\m{v}}\leq1/j$. It is obvious that $\m{v}^{\sbra{j}}\rightarrow\m{v}$. We next show that $\m{v}^{\sbra{j}}\in\cD^j$ as $j$ is large enough. By $\m{\Phi}^{\sbra{j}}\rightarrow\m{\Phi}^*$, $\m{v}^{\sbra{j}}\rightarrow\m{v}$ and that the sequence $\lbra{\m{v}^{\sbra{j}}}_{j=1}^{\infty}$ is bounded, there exists a positive integer $j_0$ such that, as $j\geq j_0$,
{\lentwo\equa{
\twon{\m{\Phi}^*-\m{\Phi}^{\sbra{j}}}\twon{\m{v}^{\sbra{j}}}&\leq&\eta/2 \notag\\
\twon{\m{\Phi}^*}\twon{\m{v}-\m{v}^{\sbra{j}}}&\leq&\eta/2.\notag}
}Hence, as $j\geq j_0$,
\[\begin{split}
&\twon{\m{y}-\m{\Phi}^{\sbra{j}}\m{v}^{\sbra{j}}}\\
=&\twon{\sbra{\m{y}-\m{\Phi}^*\m{v}} + \sbra{\m{\Phi}^*-\m{\Phi}^{\sbra{j}}}\m{v}^{\sbra{j}} + \m{\Phi}^*\sbra{\m{v}-\m{v}^{\sbra{j}}} }\\
\leq&\twon{\m{y}-\m{\Phi}^*\m{v}} + \twon{\sbra{\m{\Phi}^*-\m{\Phi}^{\sbra{j}}}\m{v}^{\sbra{j}}} + \twon{\m{\Phi}^*\sbra{\m{v}-\m{v}^{\sbra{j}}}}\\
\leq&\twon{\m{y}-\m{\Phi}^*\m{v}} + \twon{\m{\Phi}^*-\m{\Phi}^{\sbra{j}}}\twon{\m{v}^{\sbra{j}}}+ \twon{\m{\Phi}^*}\twon{\m{v}-\m{v}^{\sbra{j}}}\\
\leq&\epsilon_0+\eta/2+\eta/2=\epsilon,\end{split}\notag\]
from which we have $\m{v}^{\sbra{j}}\in\cD^j$ for $j\geq j_0$. By re-selecting arbitrary $\m{v}^{\sbra{j}}\in\cD^j$ for $j<j_0$ we obtain the conclusion.

For the other case where $\m{v}$ is a boundary point of $\cD^*$, there exists a sequence $\lbra{\m{v}_{\sbra{l}}}_{l=1}^{\infty}\subset\cD^*$ with all $\m{v}_{\sbra{l}}$ being interior points of $\cD^*$ such that $\m{v}_{\sbra{l}}\rightarrow\m{v}$, as $l\rightarrow+\infty$. According to the first part of the proof, for each $l=1,2,\cdots$, there exists a sequence $\lbra{\m{v}_{\sbra{l}}^{\sbra{j}}}_{j=1}^{\infty}$ with $\m{v}_{\sbra{l}}^{\sbra{j}}\in\cD^j$, $j=1,2,\cdots$, such that $\m{v}_{\sbra{l}}^{\sbra{j}}\rightarrow\m{v}_{\sbra{l}}$, as $j\rightarrow+\infty$. The sequence $\lbra{\m{v}_{\sbra{j}}^{\sbra{j}}}_{j=1}^{\infty}$ is what we expected since
\equ{\twon{\m{v}_{\sbra{j}}^{\sbra{j}}-\m{v}}\leq\twon{\m{v}_{\sbra{j}}^{\sbra{j}}-\m{v}_{\sbra{j}}} + \twon{\m{v}_{\sbra{j}}-\m{v}}\rightarrow0,\notag}
as $j\rightarrow+\infty$.

\section*{Appendix D\\Proof of Theorem \ref{thm:AA-P-BPDN_accumulisstation}}

We first show the existence of an accumulation point. It follows from the inequality
\equ{\twon{\m{y}-\sbra{\m{A}+\m{B}\m{\Delta}^{\sbra{j}}}\m{x}^{\sbra{j}}}\leq \twon{\m{y}-\sbra{\m{A}+\m{B}\m{\Delta}^{\sbra{j-1}}}\m{x}^{\sbra{j}}}\leq\epsilon\notag}
that $\m{x}^{\sbra{j}}$ is a feasible solution to the problem in (\ref{formu:solve_xj}), and thus $\onen{\m{x}^{\sbra{j+1}}}\leq\onen{\m{x}^{\sbra{j}}}$ for $j=1,2,\cdots$. Then we have $\onen{\m{x}^{\sbra{j}}}\leq\onen{\m{x}^{\sbra{1}}}\leq\onen{\m{A}^\dagger\m{y}}$ for $j=1,2,\cdots$, since $\m{A}^\dagger\m{y}$ is a feasible solution to the problem in (\ref{formu:solve_xj}) at the first iteration with the superscript $^\dagger$ denoting the pseudo-inverse operator. This together with $\m{\beta}^{\sbra{j}}\in\mbra{-r,r}^n$, $j=1,2,\cdots$, leads to that the sequence $\lbra{\sbra{\m{x}^{\sbra{j}},\m{\beta}^{\sbra{j}}}}_{j=1}^{\infty}$ is bounded. Thus, there exists an accumulation point $\sbra{\m{x}^*,\m{\beta}^*}$ of $\lbra{\sbra{\m{x}^{\sbra{j}},\m{\beta}^{\sbra{j}}}}_{j=1}^{\infty}$.

For the accumulation point $\sbra{\m{x}^*,\m{\beta}^*}$ there exists a subsequence $\lbra{\sbra{\m{x}^{\sbra{j_l}},\m{\beta}^{\sbra{j_l}}}}_{l=1}^{\infty}$ of $\lbra{\sbra{\m{x}^{\sbra{j}},\m{\beta}^{\sbra{j}}}}_{j=1}^{\infty}$ such that $\sbra{\m{x}^{\sbra{j_l}},\m{\beta}^{\sbra{j_l}}}\rightarrow \sbra{\m{x}^*,\m{\beta}^*}$, as $l\rightarrow+\infty$. By (\ref{formu:solve_betaj}), we have, for all $\m{\beta}\in\mbra{-r,r}^n$,
\equ{\twon{\m{y}-\sbra{\m{A}+\m{B}\m{\Delta}^{\sbra{j_l}}}\m{x}^{\sbra{j_l}}} \leq\twon{\m{y}-\sbra{\m{A}+\m{B}\m{\Delta}}\m{x}^{\sbra{j_l}}}, \notag}
at both sides of which by taking $l\rightarrow+\infty$, we have, for all $\m{\beta}\in\mbra{-r,r}^n$,
\equ{\twon{\m{y}-\sbra{\m{A}+\m{B}\m{\Delta}^*}\m{x}^*} \leq\twon{\m{y}-\sbra{\m{A}+\m{B}\m{\Delta}}\m{x}^*}, \notag}
which concludes (\ref{formu:condition_beta}).

For (\ref{formu:condition_x}), we first point out that $\onen{\m{x}^{\sbra{j}}}\rightarrow\onen{\m{x}^*}$, as $j\rightarrow+\infty$, since $\lbra{\onen{\m{x}^{\sbra{j}}}}_{j=1}^{\infty}$ is decreasing and $\m{x}^*$ is one of its accumulation points. As in Lemma \ref{Lemma:continuousnullspace}, let $\cD^j=\lbra{\m{x}: \twon{\m{y}-\sbra{\m{A}+\m{B}\m{\Delta}^{\sbra{j}}}\m{x}}\leq\epsilon}$ and $\cD^*=\lbra{\m{x}: \twon{\m{y}-\sbra{\m{A}+\m{B}\m{\Delta}^*}\m{x}}\leq\epsilon}$. By $\m{A}+\m{B}\m{\Delta}^{\sbra{j_l}}\rightarrow\m{A}+\m{B}\m{\Delta}^*$, as $l\rightarrow+\infty$, and Lemma \ref{Lemma:continuousnullspace}, for any $\m{x}\in\cD^*$ there exists a sequence $\lbra{\m{x}_{\sbra{l}}}_{l=1}^{\infty}$ with $\m{x}_{\sbra{l}}\in\cD^{j_l}$, $l=1,2,\cdots$, such that $\m{x}_{\sbra{l}}\rightarrow\m{x}$, as $l\rightarrow+\infty$. By (\ref{formu:solve_xj}), we have, for $l=1,2,\cdots$,
\equ{\onen{\m{x}^{\sbra{j_l+1}}}\leq\onen{\m{x}_{\sbra{l}}},\notag}
at both sides of which by taking $l\rightarrow+\infty$, we have
\equ{\onen{\m{x}^*}\leq\onen{\m{x}} \label{formu:formu7}}
since $\onen{\m{x}^{\sbra{j}}}\rightarrow\onen{\m{x}^*}$, as $j\rightarrow+\infty$, and $\m{x}_{\sbra{l}}\rightarrow\m{x}$, as $l\rightarrow+\infty$. Finally, (\ref{formu:condition_x}) is concluded as (\ref{formu:formu7}) holds for arbitrary $\m{x}\in\cD^*$.

\section*{Appendix E\\Proof of Theorem \ref{thm:AA-P-BPDN_optimalisstation}}

We need to show that an optimal solution $\sbra{\m{x}^*,\m{\beta}^*}$ satisfies (\ref{formu:condition_x}) and (\ref{formu:condition_beta}). It is obvious for (\ref{formu:condition_x}). For (\ref{formu:condition_beta}), we discuss two cases based on Lemma \ref{Lemma:BPDN_optimalx}. If $\twon{\m{y}}\leq\epsilon$, then $\m{x}^*=\m{0}$ and, hence, (\ref{formu:condition_beta}) holds for any $\m{\beta}^*\in\mbra{-r,r}^n$. If $\twon{\m{y}}>\epsilon$, $\twon{\m{y}-\sbra{\m{A}+\m{B}\m{\Delta}^*}\m{x}^*}=\epsilon$ holds by (\ref{formu:condition_x}) and Lemma \ref{Lemma:BPDN_optimalx}. Next we use contradiction to show that (\ref{formu:condition_beta}) holds in such case.

Suppose that (\ref{formu:condition_beta}) does not hold as $\twon{\m{y}}>\epsilon$. That is, there exists $\m{\beta}'\in\mbra{-r,r}^n$ such that
\equ{\twon{\m{y}-\sbra{\m{A}+\m{B}\m{\Delta}'}\m{x}^*}< \twon{\m{y}-\sbra{\m{A}+\m{B}\m{\Delta}^*}\m{x}^*}=\epsilon\notag}
holds with $\m{\Delta}'=\diag\sbra{\m{\beta}'}$. Then by Lemma \ref{Lemma:BPDN_optimalx} we see that $\m{x}^*$ is a feasible but not optimal solution to the problem
\equ{\min_{\m{x}}\onen{\m{x}}, \st\twon{\m{y}-\sbra{\m{A}+\m{B}\m{\Delta}'}\m{x}}\leq\epsilon.\label{formu:formu6}}
Hence, $\onen{\m{x}'}<\onen{\m{x}^*}$ holds for an optimal solution $\m{x}'$ to the problem in (\ref{formu:formu6}). Meanwhile, $\sbra{\m{x}',\m{\beta}'}$ is a feasible solution to the P-BPDN problem in (\ref{formu:l1_offgrid_Phiv}). Thus $\sbra{\m{x}^*,\m{\beta}^*}$ is not an optimal solution to the P-BPDN problem in (\ref{formu:l1_offgrid_Phiv}) by $\onen{\m{x}'}<\onen{\m{x}^*}$, which leads to contradiction.

\section*{Appendix F\\Deriavation of $\epsilon$ in Subsection \ref{sec:application}}
By (\ref{formu:offgridlinearization}), we have for $l=1,\cdots,m$, $j=1,\cdots,k$,
\equ{R_{lj}=\frac{A_{lj}''\sbra{\xi}}{2}\sbra{\theta_j-\tilde{\theta}_{l_j}}^2}
where $\xi$ is between $\theta_j$ and $\tilde{\theta}_{l_j}$, $A''_{lj}\sbra{\xi}=-\frac{\pi^2}{\sqrt{m}}\sbra{l-\frac{m+1}{2}}^2 \exp\lbra{i\pi\sbra{l-\frac{m+1}{2}}\xi}$, and $\abs{\theta_j-\tilde{\theta}_{l_j}}\leq\frac{1}{n}$. Thus, we have for $j=1,\cdots,k$,
\equ{\twon{\m{R}_j}\leq\frac{1}{2}\max\twon{\m{A}_j''}\cdot\frac{1}{n^2}= \frac{\pi^2}{8n^2}\sqrt{\frac{3m^4-10m^2+7}{15}}.}
Finally, it gives the expression of $\epsilon$ by observing that
\equ{\twon{\m{e}}=\twon{\m{R}\m{s}}\leq\frobn{\m{R}}\twon{\m{s}}\leq \sqrt{k}\twon{\m{R}_1}\twon{\m{s}}.}

\section*{Acknowledgement}
The authors would like to thank the anonymous reviewers for their valuable comments on this paper.




\end{document}